\newif\ifonecol

\onecolfalse     

\ifonecol
\documentclass[11pt,draftcls,onecolumn]{IEEEtran}
\else
\documentclass[journal]{IEEEtran}
\fi

\usepackage{pdfsync}
\usepackage{multirow}
\usepackage{amsfonts}
\usepackage{amsmath}
\usepackage{amssymb}
\usepackage[final]{graphicx}
\usepackage{xspace}
\usepackage{intmacros}
\usepackage{color}

\newcommand{\xb}{{\bf x}}
\newcommand{\vb}{{\bf v}}
\newcommand{\wb}{{\bf w}}
\newcommand{\ub}{{\bf u}}
\newcommand{\ab}{{\bf a}}

\newcommand{\nb}{{\bf n}}
\newcommand{\tnb}{{\bf \tilde n}}
\newcommand{\tsb}{{\bf \tilde s}}
\newcommand{\hsb}{{\bf \hat s}}

\newcommand{\sbb}{{\bf s}}

\newcommand{\Ab}{{\bf A}}
\newcommand{\Gb}{{\bf G}}
\newcommand{\Mb}{{\bf M}}

\newcommand{\Hb}{{\bf H}}
\newcommand{\Db}{{\bf D}}

\newcommand{\Zb}{{\bf 0}}
\newcommand{\Sb}{{\bf S}}
\newcommand{\Xb}{{\bf X}}
\newcommand{\Nb}{{\bf N}}
\newcommand{\Ib}{{\bf I}}
\newcommand{\Qb}{{\bf Q}}

\newcommand{\alphab}{\boldsymbol{\alpha}}

\newcommand{\oureq}{\Ab \sbb = \xb}
\newcommand{\Sm}{{\mathcal{S}}}
\newcommand{\Bm}{{\mathcal{B}}}
\newcommand{\Am}{{\mathcal{A}}}
\newcommand{\Cm}{{\mathcal{C}}}

\newcommand{\pinv}[1]{{#1}^\dagger}
\newcommand{\nl}{\textrm{null}}
\newcommand{\CR}{\textrm{CR}}

\newcommand{\la}{\left\{}
\newcommand{\ra}{\right\}}

\newcommand\ie{{\em i.e.}\xspace}

\newcommand{\Rr}{\mathbb{R}}

\newcommand{\Rm}{\mathbb{R}^{m}}

\newcommand{\Fs}{F_\sigma}
\newcommand{\fs}{f_\sigma}

\newcommand{\norm}[2]{\|#1\|_{#2}}

\newcommand{\nz}[1]{\|#1\|_0} 
\newcommand{\Prob}{\mathbb{P}}

\newtheorem{theorem}{Theorem}
\newtheorem{defi}{Definition}
\newtheorem{lemma}{Lemma}
\newtheorem{coro}{Corollary}

\newcommand{\inst}[1]{\unskip${^{#1}}$}

\title{}
\author{}

\title{Sparse Recovery using Smoothed $\ell^0$ (SL0): Convergence Analysis}
 \author{ Hosein Mohimani* \inst{1}, Massoud Babaie-Zadeh \inst{2}~\IEEEmembership{Senior Member, IEEE}, Irina Gorodnitsky \inst{3}~\IEEEmembership{Senior Member, IEEE}, and Christian Jutten \inst{4}~\IEEEmembership{Fellow, IEEE}
  \thanks{$^1$Department of Electrical and Computer Engineering, University of
    California, San Diego, California, USA.}
  \thanks{$^2$Department of Electrical Engineering, Sharif University of
    Technology, Tehran, Iran.}
  \thanks{$^3$Department of Cognitive Sciences, University of
    California, San Diego, California, USA.}
  \thanks{$^4$Laboratoire des Images et des Signaux (LIS),
    Institut National Polytechnique de Grenoble (INPG), France.}
  \thanks{Author's email addresses are: {\tt hmohiman@ucsd.edu}, {\tt igorodni@cogsci.ucsd.edu}, {\tt mbzadeh@sharif.edu} and {\tt Christian.Jutten@inpg.fr}}
  \ifonecol
  \thanks{Corresponding author: Massoud Babaie-Zadeh, email: {\tt mbzadeh@yahoo.com}, Tel: +98 21 66 16 59 25, Fax: +98 21 66 02 32 61.}
  \fi
  }

\begin{document}
\maketitle

\begin{abstract}
Finding the sparse solution of an underdetermined system of linear equations has many applications, especially, it is used in Compressed Sensing (CS), 
Sparse Component Analysis (SCA), and sparse decomposition of signals on overcomplete dictionaries. We have recently proposed a fast algorithm, called 
Smoothed $\ell^0$ (SL0), for this task. Contrary to many other sparse recovery algorithms, SL0 is not based on minimizing the $\ell^1$ norm, but it tries to directly 
minimize the $\ell^0$ norm of the solution. The basic idea of SL0 is optimizing a sequence of certain (continuous) cost functions approximating the $\ell^0$ norm of a vector.
However, in previous papers, we did not provide a complete convergence proof for SL0. In this paper, we study the convergence properties of SL0, and show that under a certain sparsity constraint in terms of Asymmetric Restricted Isometry Property (ARIP),
and with a certain choice of parameters, the convergence of SL0 to the sparsest solution is guaranteed.
Moreover, we study the complexity of SL0, and we show that whenever the dimension of the dictionary grows, the complexity of SL0 increases with the same order as Matching Pursuit (MP), which is one of the fastest existing sparse recovery
methods, while contrary to MP, its convergence to the sparsest solution is guaranteed under certain conditions which are satisfied through the choice of parameters.
\end{abstract}

\begin{keywords}
Compressed Sensing (CS), Sparse Component Analysis (SCA), Sparse Decomposition,
Atomic Decomposition, Over-complete Signal Representation, Sparse Source
Separation.
\end{keywords}

\ifonecol

\begin{center} \bfseries EDICS Category: MAL-SSEP, SPC-CODC, DSP-FAST \end{center}


\fi

\section{Introduction}\label{sec: int}
\IEEEPARstart{S}{parse} solution of an Underdetermined System of Linear Equations (USLE) has recently attracted the attention of many researchers from different viewpoints, because of its potential applications in many different problems. It is used, for example, in Compressed Sensing (CS)~\cite{CandRT06,Dono06,Bara07}, underdetermined Sparse Component Analysis (SCA) and source separation~\cite{GribL06, BofiZ01, GeorTC04, LiCA03}, atomic decomposition on overcomplete dictionaries~\cite{ChenDS01,DonoET06}, decoding real field codes~\cite{CandT05}, etc.

Let $\xb$ be a known $n\times1$ vector and $\Ab=[\ab_1,\dots,\ab_m]$ be a known $n \times m$ matrix with $m>n$, where $\ab_i$'s denotes its columns. Then, we can seek the sparsest solution of the USLE $\Ab \sbb = \xb$ given by
\begin{equation}
	(P_0): \quad \mbox{min } \nz{\sbb} \quad \mbox{ s.t. } \quad \Ab \sbb = \xb,
	\label{eq: P0}
\end{equation}
where $\nz{\cdot}$ is simply the number of nonzero components (conventionally called the ``$\ell^0$'' norm although it is not a true norm). In atomic decomposition viewpoint, $\xb$ is a signal which is to be decomposed as a linear combination of the signals $\ab_i$, $i=1,\dots,m$, where $\ab_i$'s are called `atoms', and $\Ab$ is called the `dictionary' over which the signal is to be decomposed \cite{MallZ93}.

A system $\Ab$ is said~\cite{GoroR97} to satisfy {\em Unique Representation Property} (URP), if any $n\times n$ sub-matrix of $\Ab$ is invertible. It is known~\cite{GoroR97,GribN03,DonoE03} that for any system satisfying URP, the solution to (\ref{eq: P0}) is unique, that is if the a solution $\sbb_0$ satisfying $\norm{\sbb_0}{0}<n/2$ exists, then any other solution $\sbb$ has $\norm{\sbb}{0}>n/2$. Therefore, under URP assumption, we can talk about `the sparsest solution'.

Solving (\ref{eq: P0}) using a combinatorial search is NP-hard. Many alternative algorithms have been proposed to solve this problem. Two frequently used approaches are Matching Pursuit (MP)~\cite{MallZ93} and Basis Pursuit (BP)~\cite{ChenDS01}, which have many variants. MP is a fast algorithm but it cannot be guaranteed to find the sparsest solution. BP is based on replacing $\ell^0$ with the $\ell^1$ norm which can be minimized using Linear Programming techniques. BP is computationally more  complex than MP, but it can find the sparsest solution with high probability, provided this solution is sufficiently sparse~\cite{GribN03, DonoE03, Dono06, DonoT06}. 

In~\cite{MohiBJ07} and~\cite{MohiBJ09}, we proposed an algorithm for solving (\ref{eq: P0}), called Smoothed $\ell^0$ (SL0), which provides a fast solution 
within a small Euclidean distance of the sparsest solution. The main idea was to approximate the $\ell^0$ norm by a smooth function (hence the name ``smoothed 
$\ell^0$''). More precisely, $\nz{\sbb}$ is approximated by a continuous function\footnote{In this form, $F_\sigma(\sbb)$ is an approximation to the number of `zero's of 
$\sbb$, that is, $m-\nz{\sbb}$.} $m-F_\sigma(\sbb)$, where $\sigma$ determines the quality of approximation: the larger $\sigma$, the smoother $F_\sigma(\cdot)$ but the worse 
the approximation to $\ell^0$; and visa versa. Hence, the 
solution  tends to the sparsest solution when $\sigma \to 0$. Therefore,  the objective underlying SL0 is to maximize $F_\sigma(\sbb)$ (subject to $\Ab \sbb =\xb$) for some {\em very small} value of $\sigma$. 
However, for small values of $\sigma$, $F_\sigma(\cdot)$ has many local maxima and hence its maximization is not easy. Therefore, SL0 uses a Graduated Non-Convexity 
(GNC)~\cite{BlakZ87} approach: It starts from a very large $\sigma$ (for which there is no local maxima), and gradually decreases $\sigma$ to zero. The maximum of 
$F_\sigma(\cdot)$ is used as a starting point to locate the maximum of $F_\sigma(\cdot)$ for the next (smaller) $\sigma$ using a steepest ascent approach. Since the 
value of $\sigma$ has only slightly decreased, the maximizer of $F_\sigma(\cdot)$ for this new $\sigma$ is not too far from the maximizer of 
$F_\sigma(\cdot)$ for the previous (larger) $\sigma$, and hence it is hoped that it does not get trapped into a local maximum. Figure~\ref{fig: SL0 alg} shows the basics of 
SL0 algorithm\footnote{Two other points in Fig.~\ref{fig: SL0 alg} are: 1) The initial guess for the sparsest solution is the minimum $\ell^2$ norm solution of 
$\Ab \sbb=\xb$, which corresponds~\cite{MohiBJ09} to the maximizer of $F_\sigma(\sbb)$ where $\sigma\to\infty$, and 2) The step-size of the steepest ascent is decreased 
proportional to $\sigma^2$ ~\cite{MohiBJ09}.}.

\begin{figure}
  \centering \vrule%
  \ifonecol
    \begin{minipage}{9.2cm} 
  \else
    \begin{minipage}{8.2cm} 
  \fi
\hrule \vspace{0.5em}
\hspace*{-1.1em}
  \ifonecol
    \begin{minipage}{9cm} 
  \else
    \begin{minipage}{8cm} 
  \fi
      {
        \footnotesize
        \def\baselinestretch{1}

        \begin{itemize}
        \item Initialization: Set $\hsb_{0}=\pinv{\Ab}\xb$. Choose a suitable decreasing sequence for $\sigma$:              $[\sigma_{1}\ldots\sigma_{J}]$.

        \item For $j=1,\dots,J$: \vspace{0.5mm}\\
        \vrule\hspace{-1em}
        \begin{minipage}{7.7cm}
          \begin{enumerate}
             \item \hspace{-0.8em} Let $\sigma=\sigma_i$.
             \item \hspace{-0.8em} Maximize $F_{\sigma}(\sbb)$ subject to $\oureq$, using $L$ iterations of steepest ascent:
             \begin{itemize}
               \item \hspace{-0.8em} Initialization: $\sbb=\hsb_{j-1}$.
               \item \hspace{-0.8em} For $\ell=1,2,\dots,L$ \vspace{0.5mm} \\
                  \hspace*{-0.4em}\vrule\hspace{-1em}
                  \begin{minipage}{8cm}
                    \begin{enumerate}
                       \item \hspace{-0.8em} Let $\sbb\leftarrow \sbb+(\mu\sigma^2) \nabla F_\sigma(\sbb)$.

                       \item \hspace{-0.8em} Project $\sbb$ back onto the feasible set $\{\sbb | \Ab \sbb = \xb\}$:
                          \begin{equation*}
                             \hspace{-8em}
                             \sbb\leftarrow\sbb-\pinv{\Ab}(\Ab\sbb-\xb).
                          \end{equation*} 
                    \end{enumerate}
                  \end{minipage}\\\hspace*{-0.4em}\rule{2mm}{.5pt}
             \end{itemize}
             \item \hspace{-0.8em} Set $\hsb_{j}=\sbb$.\vspace{0.2em}
          \end{enumerate}
        \end{minipage}\\\rule{2mm}{.5pt}

        \item Final answer is $\hsb=\hsb_J$.
        \end{itemize}
      }
    \end{minipage}
    \vspace{1em} \hrule
  \end{minipage}\vrule \\
\caption{Basics of the SL0 algorithm~\cite{MohiBJ09}. $\pinv{\Ab}$ stands for the Moore-Penrose pseudo inverse of $\Ab$ (\ie\ $\pinv{\Ab}\triangleq\Ab^{T}(\Ab\Ab^{T})^{-1}$).}

\label{fig: SL0 alg}
\end{figure}



From Fig.~\ref{fig: SL0 alg}, SL0 consists of two loops: the `outer' loop is the loop in which $\sigma$ is decreased, and the `inner' loop is the one in which $F_\sigma(\sbb)$ is iteratively maximized (subject to $\Ab \sbb=\xb$) for the {\em fixed} choice of $\sigma$. In \cite{MohiBJ09}, we prove that {\em if the inner loop does not get trapped 
in a local maximum,  our solution will converge to the solution of (\ref{eq: P0})} as $\sigma \to 0$ in the outer loop. 
In other words, if $\sigma$ is decreased so gradually that the GNC  approach works and we have avoided local maxima in the inner loop, then  our method will produce the desired results.

However, a complete convergence analysis of SL0, as well as  the choice of SL0 parameters to guarantee avoiding local maxima in the inner loops remained to be shown. In particular, we want to know 1)~the rate of decreasing of $\sigma$, 2)~how many times we have to repeat the inner loop (the value of $L$), and 3)~how to choose $\mu$ in Fig.~\ref{fig: SL0 alg}. In this paper, we present a complete convergence analysis of SL0 for both noiseless and noisy cases, and we present parameter settings that guarantee SL0  convergence to the solution of (\ref{eq: P0}). In contrast to exponential family of functions used for approximating the $\ell^0$ norm in ~\cite{MohiBJ09}, the analysis here uses a family of spline functions for this aim.


Note that, in practice, the values of SL0 parameters that guarantee the convergence to the solution of (\ref{eq: P0}) are not necessarily `good' values. These values provide a theoretical support for the SL0 algorithm, but they are often excessively pessimistic and result in slower convergence of the algorithm compared to a {\em typical} behavior (see also Section~VI of~\cite{DonoET06}).



\subsection{Restricted Isometry and Overview of the Results}

The analysis is developed here using the Asymmetric Restricted Isometry Constants (ARICs)~\cite{DaviG09, BlanCT09, FoucL09, DossPF09}, in order to relate our work to $\ell^1$-minimization. The asymmetric k-restricted constants $\delta_k^{\min}$ and $\delta_k^{\max}$ are defined as the smallest nonnegative numbers satisfying
\begin{equation}
(1-\delta_k^{\min})\norm{\sbb}{2}^2 \leq \norm{\Ab\sbb}{2}^2 \leq (1+\delta_k^{\max})\norm{\sbb}{2}^2
\label{eq: ARIP constant def}
\end{equation}
for any $\sbb \in \Rr^m$ with $\norm{\sbb}{0}\leq k$. 

Let $\sbb_0$ be the solution of (\ref{eq: P0}) and $\nz{\sbb_0}=k$. We show that SL0 recovers this solution provided that
\begin{equation}
\alpha\delta_{\lceil 2k\alpha \rceil}^{\min}+\norm{\Ab}{2} \leq \alpha
\end{equation}
for any $\alpha>1$, in which $\norm{\Ab}{2}$  denotes the Euclidean norm of $\Ab$, and $\lceil 2k\alpha \rceil$ denotes the nearest integer greater than or equal to $2k\alpha $ . More precisely, we derive a family of sufficient conditions for the performance of SL0 that depend on  parameter $\alpha$. 

The ARICs are easy to calculate exactly for small scale systems, but the complexity grows exponentially as the scale grows. In fact, the value of ARICs depends on singular values of sub-matrices of the matrix $\Ab$. Then, using the results of \cite{CandT06, CandTR06, Cand08, DaviG09, BlanCT09, FoucL09}, we analyze the behavior of SL0 for large Gaussian random dictionaries. To achieve bounds similar to the existing ones for $\ell^1$ minimization methods, we use a popular result in Random Matrix Theory \cite{Ledo01,DaviS01}, to  derive Corollary  \ref{coro: guarantee} of Section~\ref{sec: large random gaussian} which can be viewed as  SL0 counterpart of Theorem 3.1 of~\cite{Dono04}. Specifically, we identify  $\rho(\alpha)>0$, for any $0<\alpha\leq 1$, such that for large scales\footnote{By scale we mean the number of rows, $n$, and the number of columns, $m$, of the dictionary.} satisfying $n/m\rightarrow\alpha$ and    $m\rightarrow\infty$, SL0 can recover any sparse solution $\sbb$ with $\norm{\sbb}{0}<\rho(\alpha)m$ from a (possibly noisy) measurement $\xb$.

One of the bottlenecks of Compressed Sensing methods for handling large scale systems is the decoding complexity (see \cite{CandT05} for the definition of
encoding and decoding in compressed sensing context). In BP, decoding complexity is known to be $m^3$~\cite{CandT05, CandT06}, or $m^{1.5}n^2$ for the cases where $n$ is much smaller than $m$~\cite{NestNN94, BentN01}. The coding complexity 
 is $mn$. MP method has the smallest possible complexities for both encoding and decoding, which is $mn$~\cite{BlumD08}. For certain classes of systems, the complexity can be further reduced to $m\log{m}$~\cite{KrstG06}. In this paper, we will see (in Section~\ref{sec: BSS}) 
that the coding and decoding complexities of SL0 are similar to that of MP.

Since (\ref{eq: P0}) is NP-hard, one may wonder that proving convergence of SL0 (with a complexity growing in quadratic with scale) means that $\textrm{NP}=\textrm{P}$. This is not the case. Note that in BP, too, the guarantee that BP will find the solution of~(\ref{eq: P0}) does not mean that $\textrm{NP}=\textrm{P}$, because such a guarantee only exists in the case of a very sparse solution. Our analysis possesses a similar limitation, too.


The paper is organized as follows. In section~\ref{sec: convergence anal}, assuming that the internal loop of Fig.~\ref{fig: SL0 alg} exactly follows the steepest ascent trajectory (in other words, 
we ignore the effect of $\mu$ and $L$, or implicitly assume that $\mu\to 0$ and $L\to \infty$), we analyze the convergence 
of the resultant (i.e. asymptotic) SL0. Indeed, in this section, Theorem~\ref{theo: guaranteed convergence} proposes a geometric $\sigma$ sequence which guarantees the local 
concavity of cost functions and the convergence of the internal loop of SL0 to the true maximizer of $F_\sigma$, and hence the convergence of asymptotic SL0 to the sparsest 
solution. This sequence depends on the ARIP constants of the dictionary, which are not easy to calculate. Hence, in Section~\ref{sec: large random gaussian}, we discuss 
the behavior of asymptotic SL0 in the case of large random Gaussian dictionaries. Corollary~\ref{coro: guarantee} of this section corresponds Donoho's results for $\ell^1$ 
minimization, Theorem 3.1 of~\cite{Dono04}. In Section~\ref{sec: stability}, we consider the effect of no ideal $\mu$, that is, where the internal loop does not 
follow exactly the steepest ascent trajectory, and makes discrete jumps in the steepest ascent direction. We provide a choice for $\mu$ which guaranties 
stability of the internal loop and convergence to the maximizer as $L \to \infty$. Then, after a discussion on the noisy case in Section~\ref{sec: noisy}, we derive a (finite) value for $L$ in Section~\ref{sec: algorithm} which guaranties the convergence of SL0 to the sparsest solution. This completes our 
convergence analysis of SL0. Further in Section~\ref{sec: algorithm} (Theorem~\ref{theo: ultimate algorithm unknown}), we study the complexity of SL0 and prove that it is of order $O(m^2)$, that is, 
the same as for MP, which is the fastest known algorithm in the field. Finally we address multiple 
sparse solution recovery with SL0 and show that the order of complexity of SL0 can be reduced to $O(m^{1.376})$ in this case.

\section{Convergence Analysis in noiseless case} \label{sec: convergence anal}

\subsection{Basic Definitions} \label{sec: Basic Def}
In~\cite{MohiBJ09}, we first choose a continuous function $f_\sigma$ that asymptotically approximates a Kronecker delta:
\begin{equation}\label{eq: approximating prop 1}
  \lim_{\sigma \to 0} f_\sigma(s) = \left \{
  \begin{array}{ll}
    1   & ; \mbox{if $s=0$} \\
    0   & ; \mbox{if $s\neq 0$}
  \end{array}
  \right. ,
\end{equation}
and use it to approximate $\nz{\sbb}$ by $m-F_\sigma(\sbb)$ where  $F_\sigma(\sbb)\triangleq\sum_{i=1}^{m}f_\sigma(s_i)$. Then, it is shown that under some mild conditions on $f_\sigma(\cdot)$, maximizing $F_\sigma(\sbb)$ on $\Ab \sbb = \xb$ for a small $\sigma$, using a GNC approach, will recover the sparse solution. To avoid being trapped into local maxima, one may wish to design a continuous concave function $f_\sigma$ that can asymptotically approximate a Kronecker delta, but, taking into account the shape of any approximation to the Kronecker delta, this is not possible. However, we note that even for non-concave continuous functions, if the function is concave in the vicinity of the global maximum then by starting from any point sufficiently close the global maximum, steepest ascent will converge to the global maximum. In this section we investigate conditions under which $F_\sigma$ subject to $\Ab\sbb=\xb$ is concave near the global maximum, and how these can be used in designing a sequence of $\sigma$ that forces SL0 to converge to the global maximum.

{\bf Remark 1. \ } Without loss of generality, we assume that the rows of $\Ab$ are orthonormal, \ie\ $\Ab\Ab^T=\Ib_n$, where $\Ib_n$ stands for the $n\times n$ identity matrix. In effect, if the rows of $\Ab$ are not orthonormal, performing a Gram-Schmidt orthonormalization on the rows of
$\Ab$ (and doing the corresponding operations on $\xb$, too) gives rise to an equivalent system of equations with the same set of solutions and with orthonormal rows of its dictionary. 

Moreover, for any matrix $\Ab$ with orthonormal rows, by expanding the set of rows of $\Ab$
, one can find a matrix $\Db \in \Rr^{(m-n) \times m}$ such that $\Qb=[\Ab^T, \Db^T]^T$ is orthonormal. We note then that:

\begin{equation}
  \left \{
  \begin{array}{l}
    \Ab\Ab^T=\Ib_n\\
    \Db\Db^T=\Ib_{m-n}\\
    \Ab\Db^T=\Zb \\
    \Ab^T\Ab+\Db^T\Db=\Ib_m
  \end{array}
  \right. .
\end{equation}
The rows of the matrix $\Db$ are an orthonormal basis for the null-space of $\Ab$. Moreover, for any $\sbb$ satisfying $\Ab\sbb=\Zb$ we have
\begin{equation}
\label{eq: norm formula}
\norm{\Db\sbb}{}=\norm{\Qb\sbb}{}=\norm{\sbb}{},
\end{equation}
where, throughout the paper, $\|\cdot\|$ stands for the $\ell^2$ norm of a vector.

\begin{defi}
Let $\pi_i:\Rr^m\mapsto\Rr$ be the projection
of $\sbb=[s_1,\cdots,s_m]^T$ onto the $i$th axis, \ie $\pi_i(\sbb)=s_i$. Moreover, let $\pi_I(\sbb)=(s_{i_1},\cdots,s_{i_r})^T$ for $I=\{ i_1 < i_2 < \cdots < i_r \} \subseteq \{ 1 \cdots m \}$ . Also let $I^c=\{ 1 \cdots m \}-I$.
\end{defi}

{\em Example. } For $\sbb=(2,3,4,7)^T$ and $I=\{ 1 , 3\}$, we have $\pi_3(\sbb)=4$, $\pi_I(\sbb)=(2,4)^T$, and $\pi_{I^c}(\sbb)=(3,7)^T$.

\begin{defi}
For the matrix $\Ab$ we define:
\begin{equation}
\begin{split}
\gamma_{\Ab}(n_0)& \triangleq\max_{
    |I|\leq n_0} \max_{\bf{As=0}} {\frac{\norm{\pi_I(\sbb)}{}^2}{\norm{\pi_{I^c}(\sbb)}{}^2}} \\
    &= \max_{|I|\leq n_0} \max_{\bf{As=0}} {\frac{\norm{\sbb}{}^2-\norm{\pi_{I^c}(\sbb)}{}^2}{\norm{\pi_{I^c}(\sbb)}{}^2}}
    \\ &=
    \max_{|I|\leq n_0} \max_{\bf{As=0}} {\frac{\norm{\sbb}{}^2}{\norm{\pi_{I^c}(\sbb)}{}^2}}-1,
\end{split}
\label{eq: gamma def}
\end{equation}
where $|I|$ represents the cardinality of $I$. We will use $\gamma(n_0) = \gamma_{\Ab}(n_0)$ notation whenever there is no ambiguity about the matrix $\Ab$.
\end{defi}

{\bf Remark 2. \ } Let $\nl(\Ab)=\{\sbb \in \Rm | \Ab \sbb = \textbf{0}
\}$ denote the null space of $\Ab$. Then for any $\sbb \in \nl(\Ab)$:
\begin{equation}
\label{eq: norm eq}
\Ab\sbb=0 \Rightarrow \Ab_I\sbb_I+\Ab_{I^c}\sbb_{I^c}=0 \Rightarrow \norm{\Ab_I\sbb_I}{}=\norm{\Ab_{I^c}\sbb_{I^c}}{},
\end{equation}
where $\Ab_I$ and $\Ab_{I^c}$ are sub-matrices of $\Ab$ containing columns indexed by $I$ and $I^c$, respectively, $\sbb_I\triangleq\pi_I(\sbb)$ and $\sbb_{I^c}\triangleq\pi_{I^c}(\sbb)$.
Now let $\sigma_{\min}(\cdot)$ and $\sigma_{\max}(\cdot)$ stand for the smallest and largest singular values of a matrix\footnote{While it is common in the literature to define singular values to be strictly positive, in this paper, we use the definition of Horn and Johnson~\cite[pp. 414-415]{HornJ90}, in which, the number of singular values of a $p\times q$ matrix $\Mb$ is fixed equal to $\min(p,q)$, and hence, the singular values of $\Mb$ are the square roots of {\em the $\min(p,q)$ largest eigenvalues of $\Mb^H \Mb$ (or $\Mb \Mb^H$)}. Using this definition, a matrix can have zero singular values; where a zero singular value characterizes a non-full-rank matrix.}. Then from (\ref{eq: norm eq}) and
\begin{equation}
\begin{split}
\norm{\Ab_I\sbb_I}{} & \geq \sigma_{\min}(\Ab_I)\norm{\sbb_I}{}\\
\norm{\Ab_{I^c}\sbb_{I^c}}{} & \leq \sigma_{\max}(\Ab_{I^c})\norm{\sbb_{I^c}}{}
\end{split}
\end{equation}
we will have:
\begin{equation}
\label{eq: gamma ineq}
\gamma(n_0) \leq \max_{|I|\le n_0}\frac{\sigma^2_{\max}(\Ab_{I^c})}{\sigma^2_{\min}(\Ab_{I})}\cdot
\end{equation}
By a similar argument:
\begin{equation}
\label{eq: general gamma ineq}
\gamma(n_0)+1 \leq \max_{|I|\le n_0}\frac{\sigma^2_{\max}(\Ab)}{\sigma^2_{\min}(\Ab_{I})}=\frac{\norm{\Ab}{2}^2}{\min_{|I|\le n_0}\sigma^2_{\min}(\Ab_{I})}\cdot
\end{equation}
where $\norm{\cdot}{2}$ denotes the spectral norm of a matrix, that is, its largest singular value.


{\bf Remark 3. \ } $\gamma(n)<\infty$ as long as $\Ab$ satisfies the URP. Observe that for any subset $|I|\le n$, we have $\sigma^2_{\max}(\Ab_{I^c})<\infty$. 
When $\Ab$ has the URP, the columns of $\Ab_{I}$ are linearly independent as long as $|I|\leq n$, and hence $\sigma^2_{\min}(\Ab_{I})>0$. 
Then (\ref{eq: gamma ineq}) implies that $\gamma(n)$ is finite.


{\bf Remark 4. \ } $\gamma(n_0)$ is clearly an increasing function
of $n_0$.

{\bf Remark 5. \ } Our definition of $\gamma(n_0)$ in (\ref{eq: gamma def}) relates 
to the lower ARIC defined in (\ref{eq: ARIP constant def}).
From (\ref{eq: general gamma ineq}) it is easy to see that for the ARIC $\delta_k^{min}$ satisfying (\ref{eq: ARIP constant def}),
\begin{equation}
\label{eq: ARIP constant relationship}
\gamma(n_0)+1 \leq \frac{\norm{\Ab}{2}^2}{1-\delta_{n_0}^{\min}}\cdot  
\end{equation}
Considering the existing upper bounds on the ARIP constants~\cite{BlanCT09},  it is straight forward to find upper bound on $\gamma(n_0)$. We discuss the upper bound on $\gamma(n_0)$ in section \ref{sec: large random gaussian}.

{\bf Remark 6. \ } For any $n \times n$ nonsingular  matrix $\Qb$, the null spaces of $\Ab$ and $\Qb\Ab$ are equal, \ie $\{\sbb|\Ab\sbb=\Zb\} = \{\sbb|\Qb\Ab\sbb=\Zb\}$. Therefore,
  $\gamma_{\Ab}(n_0)=\gamma_{\Qb\Ab}(n_0)$ for any value of $n_0$.

{\bf Remark 7. \ } Gram Schmidt orthonormalization involves left side multiplication by a nonsingular  matrix. Therefore, it does not change the value of $\gamma$.

In~\cite{MohiBJ09}, we had used a family of Gaussian functions to approximate the $\ell^0$ norm. In this paper we use quadratic splines  instead. The second order derivative of these splines is easy to manipulate and this simplifies our convergence analysis.

\begin{defi}
Let $f_\gamma:\Rr\mapsto\Rr$ denote a quadratic spline with knots at $\{ +1,-1, 1+\gamma,-1-\gamma \}$, that is:
\begin{equation}
\label{eq: f gamma}
  f_\gamma(s) \triangleq \left \{
  \begin{array}{ll}
    1- s^2/(1+\gamma)   & ; \mbox{if $|s| \le 1$} \\
    (|s|-\gamma-1)^2/(\gamma^2+\gamma)  & ; \mbox{if $1 \le |s| \le 1+\gamma$} \\
    0 & ; \mbox{if $|s| \ge 1+\gamma$}
  \end{array}
  \right. \cdot
\end{equation}
We also define
\begin{equation}
\label{eq: f gamma sigma}
f_{\gamma,\sigma}(s)\triangleq f_{\gamma}(s/\sigma)
\end{equation}
and
\begin{equation}
\label{eq: F gamma}
F_{\gamma,\sigma}(\sbb)\triangleq\sum_{i=1}^{m} f_{\gamma,\sigma}(s).
\end{equation}
In the rest of this paper, we use the notation $F_\gamma = F_{\gamma,1}$. We also use $F_{\sigma} = F_{\gamma,\sigma}$ whenever there is no ambiguity about $\gamma$.
\end{defi}

{\bf Remark 8. \ } $f_\gamma$ and $f'_\gamma$ are both continuous, so that
\begin{equation}
\label{eq: f'}
  f'_\gamma(s) = \left \{
  \begin{array}{ll}
    -2s/(1+\gamma)   & ; \mbox{if $|s| \le 1$} \\
    2s/(\gamma^2+\gamma)-2/\gamma  & ; \mbox{if $1 \le s \le 1+\gamma$} \\
    2s/(\gamma^2+\gamma)+2/\gamma  & ; \mbox{if $-1-\gamma \le s \le -1$} \\
    0 & ; \mbox{if $|s| \ge 1+\gamma$}
  \end{array}
  \right. ,
\end{equation}
and
\begin{equation}
\label{eq: f''}
  f''_\gamma(s) = \left \{
  \begin{array}{ll}
     -2/(\gamma+1) & ; \mbox{if $|s| \le 1$} \\
     2/(\gamma^2+\gamma) & ; \mbox{if $1 \le |s| \le 1+\gamma$} \\
     0 & ; \mbox{if $|s|\ge 1+\gamma$} \\
  \end{array}
  \right. .
\end{equation}



\begin{defi}
By $\norm{\sbb}{0,\sigma}$, we mean the number of elements of $\sbb$ which have absolute values greater than $\sigma$. In other words, $\norm{\sbb}{0,\sigma}$ denotes the $\ell^0$ norm of a clipped version of $\sbb$, in which, the components with absolute values less than or equal to $\alpha$ have been clipped to zero.
\end{defi}

\subsection{ Local concavity of the cost functions} \label{sec: Local Convexity}

In this subsection, we show that $ F=F_{\gamma,\sigma}$ defined in (\ref{eq: F gamma}), with $\gamma=\gamma(n_0)$, $n_0 \leq n$, and  restricted to a certain subset 
of $\Sm_\xb\triangleq\{\sbb\in\Rm|\Ab\sbb=\xb\}$, is concave. Then, we show that this subset includes all points for which $F > n_0/(1+\gamma)$.

\vspace{0.5em}
\begin{lemma}\label{lemma: negdef}
Let’s denote $F=F_{\gamma,\sigma}$, where $\gamma=\gamma(n_0)$ for $n_0 \leq n$, and have $\Ab$ satisfy the URP.
Let $\Sm_\xb\triangleq\{\sbb \in
\Rm | \Ab \sbb = \xb \}$ and $\Cm$ be the subset of $\Sm_\xb$ consisting of those solutions that have at most $n_0$ elements with
absolute values greater than $\sigma$, that is:
\begin{equation}
\label{eq: Cm def}
\Cm\triangleq\{\sbb\in\Sm_{\xb}|\, \norm{\sbb}{0,\sigma}\le n_0 \} \cdot
\end{equation}
Then the Hessian matrix of $F|_{\Cm}$, where
$F|_{\Cm}$ denotes the restriction of $F$ on $\Cm$, is negative semi-definite.
\end{lemma}

\begin{proof}
Let the linear transformation $T:\Rr^{m-n}\mapsto \Sm_\xb$ defined by
$\sbb=T(\vb)\triangleq\Db^T\vb+\Ab^T\xb$ for a constant $\xb$. $T$ is clearly a linear
isomorphism. Hence, instead of showing that the Hessian of $F|_{\Cm}$ 
is negative semi-definite, 
we just need to show that the Hessian of
$G$ is negative semi-definite on $T^{-1}(\Cm)\subseteq\Rr^{m-n}$, where $G=F \circ T$.

Assume $\sbb\in\Cm$. Clearly
\begin{displaymath}
\Hb_G(\vb)=\Db \Hb_F(\sbb)\Db^T,
\end{displaymath}
where $\vb=T^{-1}{(\sbb)}$ and
\begin{displaymath}
\begin{split}
\Hb_F(\sbb)&=\diag\big(f''_{\gamma,\sigma}(s_1), \dots,f''_{\gamma,\sigma}(s_m)\big) \\
&=\frac{1}{\sigma^2}\diag\big(f''_{\gamma}(s_1/\sigma), \dots,f''_{\gamma}(s_m/\sigma)\big) \cdot
\end{split}
\end{displaymath}
Let $I$ be the set of indexes of those elements of $\sbb\in\Cm$ that have absolute values greater
than $\sigma$. From the definition of $\Cm$, $|I|\le n_0$. To prove that $\Hb_G(\vb)$ is
negative semi-definite, we have to show that $\ub^T \Db^T \Hb_F(\sbb) \Db \ub \le 0$ for all $\ub \in \Rr^m$. Defining $\wb\triangleq\Db\ub$, we have $\Ab\wb=\Ab\Db\ub=\Zb$ and, therefore, $\wb\in \nl(\Ab)$. Next we show that $\wb^T\Hb_F(\sbb)\wb \le 0$ for all $\wb \in \nl(\Ab)$.
We write:
\begin{equation}
\label{eq: square sum}
\begin{split}
\wb^T \Hb_F(\sbb) \wb&=\frac{1}{\sigma^2}\sum_{i=1}^m
f_{\gamma}''(s_i/\sigma){w_i}^2\\&=\frac{1}{\sigma^2}\sum_{i\in I} f_{\gamma}''(s_i/\sigma){w_i}^2+\frac{1}{\sigma^2}\sum_{i\notin I}
f_{\gamma}''(s_i/\sigma){w_i}^2.
\end{split}
\end{equation}
By setting $\wb_{I}$ and $\wb_{I^c}$ equal to the sub-vectors of
$\wb$ indexed by $I$ and $I^c$, from (\ref{eq: gamma def}) we have
\begin{equation}
\label{eq: gamma formula}
\frac{\norm{\wb_I}{}^2}{\norm{\wb_{I^c}}{}^2}=\frac{\norm{\pi_I(\wb)}{}^2}{\norm{\pi_{I^c}(\wb)}{2}^2}\leq \gamma(|I|) \leq
\gamma(n_0)=\gamma,
\end{equation}
and hence, using (\ref{eq: f''}):
\begin{displaymath}
\wb^T \Hb_F(\sbb) \wb \le
-\frac{2}{1+\gamma}\frac{\norm{\wb_{I^c}}{}^2}{\sigma^2}+\frac{2}{\gamma^2+\gamma}\frac{\norm{\wb_{I}}{}^2}{\sigma^2}
\le 0,
\end{displaymath}
which completes the proof.
\end{proof}

\medskip

\begin{coro}\label{coro: convexity}
Under the conditions of Lemma~\ref{lemma: negdef}, $F=F_{\gamma,\sigma}$ is
concave at every $\sbb\in\Bm$, where $\Bm\triangleq\{\sbb \in \Sm_\xb |
F(\sbb)\ge m - n_0/(1+\gamma)\}$. Moreover, the region $\Am\triangleq\{ \sbb
\in \Sm_\xb | F(\sbb)\ge m-n_0/(2+2\gamma)\}\subseteq \Bm$ is convex.
\end{coro}

\begin{proof}
To prove the first part we show that $\Bm \subseteq \Cm$, where $\Cm$ is defined by \ref{eq: Cm def}. Let $\sbb\in\Bm$ and $I\triangleq\{ 1 \le i \le m \,|\, |s_i|>\sigma \}$. Then $\norm{\sbb}{0,\sigma}=|I|$, and hence, to prove $\sbb\in\Cm$ we have to show that $|I| \leq n_0$. We write:
\begin{gather}
    \forall i\in I: |s_i|\geq \sigma \Rightarrow 1-f(s_i) \geq 1/(1+\gamma) \label{eq: coro1 temp1}\\
    \begin{split}
    \sbb \in \Bm \Rightarrow \frac{n_0}{1+\gamma} &\geq m-F(\sbb)=\sum_{i=1}^{m} \{1-f(s_i)\} \\
    & \geq \sum_{i \in I} \{1-f(s_i)\}.
    \end{split}
    \label{eq: coro1 temp2}
\end{gather}
Substituting (\ref{eq: coro1 temp1}) in (\ref{eq: coro1 temp2}), we obtain $n_0/(1+\gamma)\ge |I|/(1+\gamma)$, which completes the proof of the first part.

To prove the second part, we consider $\sbb_1,\sbb_2\in\Am$. By definition, at most $\frac{n_0}{2}$
elements of $\sbb_1$ and $\sbb_2$ can be greater than $\sigma$. Hence, if
we define $\sbb(t)=(1-t)\sbb_1+t\sbb_2$, for $0 \leq t \leq 1$, at
most $n_0$ elements of $\sbb(t)$ can have absolute values greater than
$\sigma$. We know $\dot{\sbb}(t) = \sbb_2-\sbb_1 \in \nl(\Ab)$, and the hessian of $F$ is negative semi-definite on $\nl(\Ab)$ according to
Lemma~\ref{lemma: negdef}. Hence, if we define $h(t)=F(\sbb(t))$, we obtain:
\begin{displaymath}
\ddot{h}=\dot{\sbb}^T \Hb_F\dot{\sbb}+(\nabla F)^T
\ddot{\sbb}=\dot{\sbb}^T \Hb_F\dot{\sbb}\leq 0 \cdot
\end{displaymath}
Hence, $h$ is concave on the $[0,1]$ interval, and for any $0 \leq t \leq 1$, we have
\begin{displaymath}
F(\sbb(t))=h(t)\geq t\cdot h(1)+ (1-t)\cdot h(0) \geq m-n_0/(2+2\gamma)
\end{displaymath}. This implies that $\sbb(t)\in\Am$, hence $\Am$ is convex. 
\end{proof}

\vspace{0.5em}
\begin{coro}
\label{lemma: steepest ascent}
Under the conditions of Lemma~\ref{lemma: negdef} and the assumption that there exists a sparse solution $\sbb_0$ satisfying $k\triangleq\norm{\sbb_0}{0} \le n_0/(2+2\gamma)$, by starting from any $\hat{\sbb}$ satisfying
$F(\hat{\sbb})\geq m-n_0/(2+2\gamma)$ and moving on the steepest
ascent trajectory restricted to $\Sm_\xb$, we reach the global maximum $\sbb_*$ of $F|_{\Sm_\xb}$,
satisfying $F(\sbb_*)\geq m-k$. More precisely, the solution of the differential equation
\begin{equation}
\left \{
 \begin{array}{l}
 \dot{\alphab}(t)=\nabla F|_{\Sm_\xb}\\
 \alphab(0)=\hat{\sbb}
 \end{array}
\right.
\end{equation}
satisfies
\begin{equation}
\lim_{t \rightarrow + \infty}\alphab(t)=\sbb_* \,.
\end{equation}

\end{coro}

\begin{proof}
From Corollary \ref{coro: convexity} we know that $\Am=\{\sbb|F(\sbb)\geq
m-n_0/(2+2\gamma)\}$ is a convex region. By starting from any point
in a convex region and moving on the steepest ascent trajectory of a function
which is concave on that region, we achieve the global maximizer in that region. Therefore, the steepest ascent trajectory leads to the maximizer
$\sbb_*\in\Am$. Using the assumptions on sparse solution, we have $\sbb_0\in \Am$. Hence the maximizer clearly satisfies $F(\sbb_*)\geq
F(\sbb_0)\geq m-k$.
\end{proof}


\subsection{ The narrow  variation property} \label{sec: Uniqueness of Maximum}
In this subsection, we introduce a notion of the narrow variation property,  which states that whenever the values of $F_\sigma$ at two points exceed a certain threshold,
those two points are close to each other in the sense of the Euclidean distance between them being  bounded by $O(m^{1/2} \gamma^{1/2} \sigma)$. Before 
stating Lemma~\ref{lemma: L0 solution approximation},
we repeat Theorem~1 from \cite{MohiBJ09}. This theorem states that
if for each value of $\sigma$ we pick a point $\sbb_\sigma$ on $\Sm_\xb$ such that $F_\sigma(\sbb_\sigma)$ is greater than a certain value $m-n+k$, then the sequence of 
these points converges to the sparsest solution as $\sigma \to 0$.

\vspace{0.5em}
\begin{theorem}
\label{theo: solution convergence} Consider a family of univariate
functions $\fs$, indexed by $\sigma$, $\sigma \in \Rr^{+}$,
satisfying the set of conditions:
\begin{enumerate}
\item $\lim_{\sigma \to 0} \fs(s)=0 \qquad \textrm{; for all }s \neq 0$
\item ${\fs(0)}=1 \qquad \textrm{; for all }\sigma \in \Rr^{+}$
\item $0 \leq \fs(s) \leq 1 \qquad \textrm{; for all }\sigma \in \Rr^{+}, s \in \Rr$
\item For each positive values of $\nu$ and $\alpha$, there exists
$\sigma_{0} \in \Rr^{+}$ that satisfies: \label{item: cond4}
\begin{equation}
|s|>\alpha \Rightarrow \fs(s)<\nu \quad \textrm{; for all
}\sigma<\sigma_{0}. \label{eq: cond 4 main theorem}
\end{equation}
\end{enumerate}
Let $\Fs(\sbb)\triangleq\sum_{i=1}^{n}{\fs(s_i)}$. Assume
that $\Ab$ satisfies the URP, $\sbb_0\in\Sm_\xb$
satisfies ${\norm{\sbb_0}{0}}=k\leq n/2$ and $\sbb_\sigma\in\Sm_\xb$
satisfies $F_\sigma(\sbb_\sigma)\geq m-n+k$. Then
\begin{equation}
\lim_{\sigma \to 0}{\sbb_\sigma}=\sbb_{0}. \label{eq: s-sig to s0}
\end{equation}
\end{theorem}

{\bf Remark 1. \ } Note that the conditions on $\Ab$ in Lemma~\ref{lemma:
negdef} are the same as in Theorem~\ref{theo: solution convergence},
and $f_{\gamma,\sigma}$ defined in (\ref{eq: f
gamma sigma}) satisfies all the conditions 1 to 4 of
Theorem~\ref{theo: solution convergence}, for any arbitrary value of $\gamma$.

The main idea of the following Lemma~\ref{lemma: L0 solution approximation} (and its proof) is very similar to that of Theorem \ref{theo: solution convergence}. We prove that if 
$F_{\gamma,\sigma}$ values at two points $\sbb_1$ and $\sbb_2$ in $\Sm_\xb$ are larger than $m-n_0/(2+2\gamma)$, then
the distance between $\sbb_1$ and $\sbb_2$ is bounded by $2\sqrt{m(\gamma+1)}\sigma$.

\begin{lemma}
\label{lemma: L0 solution approximation}
Let $F=F_{\gamma,\sigma}$ where $\gamma=\gamma(n_0)$. If for two points $\sbb_1$ and $\sbb_2$ of $\Sm_\xb$ we have:
\begin{equation}
\label{eq: condition}
F(\sbb_i)\geq m - \frac{n_0}{2+2\gamma}, \quad i=1,2,
\end{equation}
then:
\begin{equation}
\label{eq: narrowness bound}
\norm{\sbb_1-\sbb_2}{}\leq 2\sqrt{m(\gamma+1)}\sigma.
\end{equation}
Moreover, if $\sbb_2=\sbb_0$, we have a slightly stricter bound
\begin{equation}
\norm{\sbb_1-\sbb_0}{}\leq \sqrt{m(\gamma+1)}\sigma \cdot
\end{equation}
\end{lemma}
\begin{proof}
The argument is similar to that of Lemma~1 of~\cite{MohiBJ09}, but made a bit more rigorous. Having in mind the proof of the first part of Corollary~\ref{coro: convexity}, observe that (\ref{eq: condition}) implies that $\sbb_1$ and $\sbb_2$ have at most $n_0/2$ elements with absolute values greater than $\sigma$. Hence, $\sbb_1-\sbb_2$ has at most $n_0$ elements with absolute values greater than $2\sigma$. Let $I$ index those elements of $\sbb_1-\sbb_2$ with absolute values greater than $2\sigma$. Then $|I|\leq n_0$ and
\begin{equation}
\label{eq: pi Ic}
\norm{\pi_{I^c}(\sbb_1-\sbb_2)}{}^2\leq |I^c|(2\sigma)^2\leq 4m\sigma^2.
\end{equation}
From (\ref{eq: gamma formula}) and (\ref{eq: pi Ic}), we get
\begin{equation}
\norm{\pi_{I}(\sbb_1-\sbb_2)}{}^2\leq 4m\sigma^2\gamma 
\end{equation}
and
\begin{equation}
\norm{\sbb_1-\sbb_2}{}^2\leq 4m\sigma^2(1+\gamma),
\end{equation}
which yield (\ref{eq: narrowness bound}). If $\sbb_2=\sbb_0$, we can conclude that $\sbb_1-\sbb_2$ has at most $n_0$ elements with absolute values greater than $\sigma$, and hence
\begin{equation}
\norm{\sbb_1-\sbb_0}{}^2\leq m\sigma^2(1+\gamma).
\end{equation}
\end{proof}

\subsection{ Bounded variations of cost functions} \label{sec: Bounded Variation}

Our cost functions have a nice property which $\ell^0$ does not, \ie they are continuous. In Lemma~\ref{lemma: Derivation} we show that the derivative of $f$ is bounded, and as a result, small changes in $\sbb$ result in small changes in $F(\sbb)$.
\begin{lemma}
\label{lemma: Derivation}
For $f=f_{\gamma,\sigma}$ and $F=F_{\gamma,\sigma}$:
\begin{equation}
\label{eq: derivation}
|f'(s)|<\frac{2}{(1+\gamma)\sigma}
\end{equation}
and
\begin{equation}
\label{eq: deviation}
|F(\sbb_1)-F(\sbb_2)|\leq \frac{2\sqrt{m}}{(1+\gamma)\sigma}\norm{\sbb_1-\sbb_2}{}
\end{equation}
for any $s\in\Rr$ and $\sbb_1,\sbb_2 \in \Rr^m$.
\end{lemma}
\begin{proof}
(\ref{eq: derivation}) is a straight forward conclusion from (\ref{eq: f'}). To prove (\ref{eq: deviation}), note that for any $\sbb\in\Rr^m$ we have
\begin{equation}
\norm{\nabla F(\sbb)}{2}=\sqrt{\sum_{i=1}^{m}|f'(s_i)|^2}\leq \frac{2\sqrt{m}}{(1+\gamma)\sigma},
\end{equation}
where $\nabla F$ denotes the gradient of $F:\Rm\to\Rr$. Moreover, using the mean value theorem, for any $\sbb_1$ and $\sbb_2$ there exists a $\sbb \in\Rr^m$ such that
\begin{equation}
F(\sbb_1)-F(\sbb_2)=\nabla F(\sbb)^T(\sbb_1-\sbb_2)
\end{equation}
Therefore:
\begin{equation}
\begin{split}
|F(\sbb_1)&-F(\sbb_2)|=|\nabla F(\sbb)^T(\sbb_1-\sbb_2)| \\& \leq \norm{\nabla F(\sbb)}{2}\cdot\norm{\sbb_1-\sbb_2}{} \leq \frac{2\sqrt{m}}{(1+\gamma)\sigma}\norm{\sbb_1-\sbb_2}{} \cdot
\end{split}
\end{equation}
\end{proof}

\subsection{ The choice of parameters of the algorithm} \label{sec: Parameter Choice}

At this point we have acquired the necessary tools for designing a sequence of $\sigma$ values needed to successfully maximize $F_{\gamma,\sigma}$. 
The question remained to be solved is how, after finding the global maximum of $F_{\gamma, \sigma}$ for some value of $\sigma$, we choose the next value of $\sigma$ so that we are 
guaranteed to be in a (locally) concave area. More specifically, Lemma~\ref{lemma: steepest ascent} ensures that by starting from any point $\sbb$ satisfying 
$F_{\gamma,\sigma}(\sbb) \geq m-n_0/(2+2\gamma)$ and following the steepest ascent trajectory of $F_{\gamma,\sigma}$, we end at the global maximum $\sbb_*$ 
of $F_{\gamma,\sigma}$ satisfying $F_{\gamma,\sigma}(\sbb_*)\geq m-k$. The question we study next is how, knowing $F_{\gamma,\sigma}(\sbb_*) \geq m-k$, can we choose 
the next value of $\sigma'$ subject to $F_{\gamma,\sigma'}(\sbb_*)\geq m-n_0/(2+2\gamma)$.
In Lemma~\ref{lemma: choose sigma 1 and c} we present a constant $c$, for which $\sigma'=c\sigma$ satisfies this condition.

\begin{lemma}
\label{lemma: choose sigma 1 and c}
For constants $B\geq A \geq 0$, let's define
\begin{equation}
\label{eq: c choice}
c\triangleq\frac{2m}{2m+B-A}\cdot
\end{equation}
Then we have the following result:
\begin{equation}
\label{eq: c condition}
\textrm{If $F_{\gamma,\sigma}(\sbb) \geq m-A$, then $F_{\gamma,c\sigma}(\sbb) \geq m-B$,}
\end{equation}
for any $\sbb\in\Rr^m$. Moreover
\begin{equation}
\label{eq: F inequality}
F_{\gamma,\sigma}(\sbb) \geq m - \frac{\norm{\sbb}{}^2}{(1+\gamma)\sigma^2}\cdot
\end{equation}
\end{lemma}

\begin{proof}
For (\ref{eq: F inequality}) note that:
\begin{displaymath}
f_{\gamma}(s/\sigma)\geq 1-s^2/(1+\gamma)\sigma^2 \Rightarrow F_{\gamma,\sigma}(\sbb)\geq
m-\frac{\norm{\sbb}{}^2}{(1+\gamma)\sigma^2}\cdot
\end{displaymath}

\noindent Let's define:
\begin{displaymath}
\alpha(t)\triangleq F_{\gamma,\sigma/(1+t)}{(\sbb)}=F_{\gamma,\sigma}{(\sbb+\sbb t)}=\sum_{i=1}^{n}{f_{\gamma,\sigma}(s_i+s_i t)}
\end{displaymath}
for $t\geq 0$. Having $|f'_{\gamma,\sigma}(s)|\leq 2/(1+\gamma)\sigma$ from (\ref{eq: derivation}), and $f'_{\gamma,\sigma}(s)=0$ for $|s| \geq (1+\gamma)\sigma$ from (\ref{eq: f'}), we will have
\begin{multline*}
|\frac{d}{dt}\alpha(t)|=|\sum_{i=1}^{m} \frac{d}{dt} f_{\gamma,\sigma}(s_i+s_i t)| \leq \sum_{i=1}^{m} |s_i|\cdot|f_{\gamma,\sigma}'(s_i+s_i t))| \\ = \sum_{|s_i|<\sigma(1+\gamma)}|s_i|\cdot|f_{\gamma,\sigma}'(s_i+s_i t))|
\leq 2m \cdot
\end{multline*}
Hence, by choosing $t_0=(B-A)/(2m)$, we have
\begin{displaymath}
|\alpha(t_0)-\alpha(0)|\leq t_0|\frac{d}{dt}\alpha(t)|\leq B-A
\end{displaymath}
for some $t \geq 0$. Then, choosing $c=1/(1+t_0)$ in (\ref{eq: c choice}), we have
\begin{equation}
|F_{c\sigma}(\sbb)-F_{\sigma}(\sbb)| = |\alpha(t_0)-\alpha(0)| \leq B-A,
\end{equation}
which leads to (\ref{eq: c condition}).
\end{proof}

Using Lemma~\ref{lemma: choose sigma 1 and c}, the following theorem states a sufficient condition for the convergence of an asymptotic version of SL0, in which the steepest ascent follows 
exactly the steepest ascent trajectory (\ie the case $\mu\to0$ and $L \to \infty$).

\begin{theorem}
\label{theo: guaranteed convergence} Assume $\Ab$ satisfies the URP and
$f$ is as defined in (\ref{eq: f gamma}), and also $k\triangleq \nz{\sbb_0} <n_0/(2+2\gamma)$. Let
$\hsb\triangleq\argmin_{\oureq}{\norm{\sbb}{}}=\pinv{\Ab}\xb$
and:
\begin{equation}
\label{eq: sigma 1}
\sigma_1=\frac{\norm{\hsb}{}}{\sqrt{k(1+\gamma)}}
\end{equation}
\begin{equation}
\label{eq: c}
c=\frac{2m}{2m+n_0/(2+2\gamma)-k}<1 \cdot
\end{equation}
If we choose the geometric sequence of $\sigma$ according to
$\sigma_{j+1}=c\sigma_j$, and set $\sbb_1=\hsb$ in the first step, and
in each subsequent step, \ie $j \geq 2$, start with $\sbb_{j-1}$ and move on the
steepest ascent trajectory of $F_{\sigma_j}$ to reach the maximizer
$\sbb_{j}$, then at each step:
\begin{displaymath}
F_{\sigma_j}(\sbb_j)\geq m-k
\end{displaymath}
and
\begin{displaymath}
\lim_{j \to \infty}{\sbb_j}=\sbb_{0}.
\end{displaymath}
\end{theorem}

\medskip

\begin{proof}
By induction on $j$. First note that by substituting $\sigma_1$ defined by (\ref{eq: sigma 1}) in (\ref{eq: F inequality}), we have
$F_{\sigma_1}(\sbb_1)=F_{\sigma_1}(\hsb)\geq m-k$. Moreover, by substituting $c$ defined by (\ref{eq: c}) in Lemma~\ref{lemma: choose sigma 1 and c}
we conclude
\begin{equation}
\label{eq: lem 5 conclusion}
F_{\sigma}(\sbb)\geq m-k\Rightarrow F_{c\sigma}(\sbb)\geq
m-\frac{n_0}{2+2\gamma},
\end{equation}
for any $\sbb \in \Rm$. Now, to complete the induction, assume $F_{\sigma_{j-1}}(\sbb_{j-1})\geq
m-k$. Then, from (\ref{eq: lem 5 conclusion})
\begin{gather}
F_{\sigma_{j}}(\sbb_{j-1})=F_{c\sigma_{j-1}}(\sbb_{j-1}) \geq
m-\frac{n_0}{2+2\gamma}\cdot
\end{gather}
Therefore, according to Lemma~\ref{lemma: steepest ascent}, the $\sbb_{j}$
which is achieved by starting at $\sbb_{j-1}$ and following the steepest ascent trajectory of
$F_{\sigma_j}$, satisfies $F_{\sigma_j}(\sbb_j)\geq
m-k$.

To prove the second part of Theorem \ref{theo: guaranteed convergence}, note that $\sigma_j \rightarrow
0$ as $j\rightarrow\infty$ (since $c<1$) and $m-k\geq m-n+k$ (since $k\leq
n_0/2\leq n/2$), hence the sequence of $\sbb_j$ satisfies the conditions
of Theorem~\ref{theo: solution convergence}. The same conclusion also follows from Lemma \ref{lemma: L0 solution approximation}, 
since $F_{\sigma}(\sbb_j)\geq m-k \geq m - n_0/(2+2\gamma)$ results in
\begin{equation}
\norm{\sbb_j-\sbb_0}{}\leq \sqrt{m(\gamma+1)}\sigma_j\rightarrow 0 \cdot
\end{equation}
\end{proof}

{\bf Remark 1. \ } Theorem~\ref{theo: guaranteed convergence} proves the convergence of an asymptotic version of SL0, in which the internal loop steps precisely along the 
steepest ascent trajectory. This corresponds to $\mu \to 0$ and $L \to \infty$ in Fig.~\ref{fig: SL0 alg}. We will discuss later in Section~\ref{sec: stability} the case of  $\mu >0$ (discrete steps in the steepest ascent directions), and propose a value for $\mu$ which guarantees the convergence, provided that the internal loop 
is repeated until the convergence is achieved (corresponding to $L \to \infty$). Finally, Section~\ref{sec: algorithm} proposes  a value for $L$ that guarantees the convergence and that completes the convergence analysis of SL0.

{\bf Remark 2. \ } In \cite{MohiBJ09} (Remark 5, section III) we heuristically justified that $\sigma_1$ should be chosen proportional to the maximum absolute value of 
elements of $\sbb$, \ie $\max_i{|s_i|}$. This choice is now better justified by Theorem \ref{theo: guaranteed convergence}, Eq. (\ref{eq: sigma 1}).

{\bf Remark 3. \ } In Experiment 2 of \cite{MohiBJ09} we had observed that the value of $c$ depended on the sparsity ($k$) of the solution, and not as much on any other 
parameter of SL0 (see Fig.~3 of~\cite{MohiBJ09}). The optimal value of $c$ grew with increasing $k$ and tended to 1 as $k \rightarrow n/2$. Equation (\ref{eq: c}) supports this observation, as the value of $c$ depends only on the value of $k$ (and of course, the system scale), and $c \to 1$ as $k \rightarrow n_0/2(1+\gamma)$.

\begin{coro}
\label{coro: sufficient condition}
Asymptotic SL0 (when $\mu \to 0$ and $L \to \infty$) converges to the sparse solution if
\begin{equation}
\label{eq: sufficient condition}
\alpha\delta_{\lceil 2k\alpha \rceil}^{\min}+\norm{\Ab}{2} \leq \alpha,
\end{equation}
where $k=\norm{\sbb}{0}$, $\alpha>1$ is an arbitrary constant, $\delta_{k}^{\min}$ is the lower 
ARIC, and $\norm{\Ab}{2}$ denotes the spectral norm of $\Ab$.
\end{coro}
\begin{proof}
If (\ref{eq: sufficient condition}) holds, by setting $n_0=\lceil 2\alpha k \rceil$ and using (\ref{eq: ARIP constant relationship}), it is easy to see that $\gamma(n_0)+1<\alpha$.
Hence, the condition of Theorem \ref{theo: guaranteed convergence}, \ie\ $\norm{\sbb}{0}<n_0/(2+2
\gamma(n_0))$, holds and the convergence is guaranteed.

\end{proof}

\section{ Large Random Gaussian Matricies } \label{sec: large random gaussian}

Our sparsity constraint for successful recovery of the sparse solution is of the form $k<n_0/(2+2\gamma)$, where $\gamma=\gamma(n_0)$ depends on the matrix $\Ab$. 
It is not practical to precisely calculate $\gamma(n_0)$ for large scale systems since computational complexity grows exponentially\footnote{Even a deterministic upper bound on $\gamma$ using  (\ref{eq: ARIP constant relationship}) is not practical. The upper bound depends on Euclidean norm of $\Ab$ and the lower ARIC. Precise calculaion of ARIC requires enumerating all possible $n_0$-column submatrices of $\Ab$ and computing their smallest singular values.
}. However, in the case of random Gaussian matrices we can find reasonable almost sure (a.s.) upper-bounds on $\gamma(n_0)$, which make it possible to compare our results with  the ones for $\ell^1$-minimization
\cite{Dono04, CandT06, CandTR06, Cand08}. In this section we assume that $\Ab$ has independent identically distributed (i.i.d) entries drawn from a normal 
distribution with zero mean and variance $1/n$.

We use Theorem II.13 of \cite{DaviS01}. Let $\Gb$ be an $l \times n$ random matrix with i.i.d.~entries drawn from a $N(0,1/n)$ distribution. We are interested in singular 
values of $\Gb$, or equivalently, eigenvalues of $\Gb^{T}\Gb$, and, in particular, the smallest and the largest one. In \cite{DaviS01,Ledo01}, authors prove that
\begin{equation}
\label{eq: sig max}
\Prob\la\sigma_{\max}(\Gb)>1+\sqrt{l/n}+r\ra\leq \exp(-nr^2/2)
\end{equation}
and also
\begin{equation}
\label{eq: sig min}
\Prob\la\sigma_{\min}(\Gb)<1-\sqrt{l/n}-r\ra\leq \exp(-nr^2/2)\cdot
\end{equation}
They prove the above inequalities for the case $l \leq n$.
It is not difficult to check that (\ref{eq: sig max}) holds for the case $l>n$ as well, since from definition of $\Gb$, $\sqrt{n/l}\,\Gb^T$ is an $n \times l$ normal 
distributed matrix with variance $1/l$. In this case, we can use (\ref{eq: sig max}) to conclude that
\begin{equation*}
\Prob\la\sigma_{\max}(\sqrt{n/l} \,\Gb^T)>1+\sqrt{n/l}+r\ra\leq \exp(-lr^2/2)\cdot
\end{equation*}
Noting that $\sigma_{\max}(\sqrt{n/l} \, \Gb^T)=\sqrt{n/l}\,\sigma_{\max}(\Gb)$ and setting $r'=r\sqrt{l/n}$ we get the desired result.

In the following theorem, using arguments similar to ones used for bounding the symmetric and asymmetric RICs \cite{CandT06, CandTR06, Cand08,BlanCT09}, we prove that with high probability the value of $\gamma$ is bounded.

\begin{theorem}
\label{theo: large gaussian case}
If $\Ab$ is a random Gaussian matrix with i.i.d.~zero mean entries of variance $1/n$ and if $\alpha=n/m$ and $\beta=n_0/m$ are fixed, then
\begin{multline}
\label{eq: gamma prob}
\Prob\la\gamma(n_0)>\frac{(1+\sqrt{1/\alpha}+\epsilon)^2}{(1-\sqrt{\beta/\alpha}-r)^2}\ra \leq \\ \exp(-nr^2/2+nr_0^2/2)+\exp(-n \epsilon^2/2),
\end{multline}
which tends to zero
as $m \to \infty$, provided that $\epsilon>0$ and $r>r_0$ where
\begin{equation}
\label{eq: r_0}
r_0\triangleq\sqrt{2\beta/\alpha\log(\e/\beta)}
\end{equation}
and $\e=\exp(1)$ denotes the Euler's constant (the base of natural logarithm).
\end{theorem}

\begin{proof}
Let $I$ be some subset of $\{1,\cdots,m\}$ with $|I|=n_0$. Then, $\Ab_{I}$  is $n_0 \times n$ and
\begin{equation}
\Prob\la\sigma_{\max}(\Ab)>1+\sqrt{m/n}+\epsilon\ra\leq \exp(-n\epsilon^2/2)
\end{equation}
and
\begin{equation}
\Prob\la\sigma_{\min}(\Ab_{I})<1-\sqrt{n_0/n}-r\ra\leq \exp(-nr^2/2)
\end{equation}
for any subset $|I|=n_0$. There are a total of $\binom{m}{n_0}$ such subsets, which means
\begin{equation}
\Prob\la\min_{|I|=n_0}{\sigma_{\min}(\Ab_{I})}<1-\sqrt{n_0/n}-r\ra
\leq \binom{m}{n_0} \e^{-nr^2/2}. 
\end{equation}
Then, using (\ref{eq: general gamma ineq}) we have
\begin{multline}
\Prob\la\gamma(n_0)>\frac{(1+\sqrt{m/n}+\epsilon)^2}{(1-\sqrt{n_0/n}-r)^2}\ra \leq \\ \binom{m}{n_0} \exp(-nr^2/2) + \exp(-n\epsilon^2/2) \cdot
\end{multline}
 From
\begin{equation}
\binom{m}{n_0}\leq \Big( \frac{m\e}{n_0} \Big)^{n_0}\leq \exp \Big(n_0\log(m\e/n_0)\Big)
\end{equation}
we get
\begin{multline}
\Prob\la\gamma(n_0)>\frac{(1+\sqrt{m/n}+\epsilon)^2}{(1-\sqrt{n_0/n}-r)^2}\ra\leq \\ \exp \Big(n_0\log(m\e/n_0)-nr^2/2\Big)+\exp(-n\epsilon^2/2) \cdot
\end{multline}
If we assume $\alpha=n/m$ and $\beta=n_0/m$ are fixed, then by defining $r_0$ as in (\ref{eq: r_0}), we obtain (\ref{eq: gamma prob}) as $m\rightarrow\infty$.
\end{proof}

\begin{coro}
\label{coro: guarantee}
Let's define $\gamma(\alpha,\beta)$ as follows:
\begin{equation*}
	\gamma(\alpha,\beta) \triangleq
	 \frac{(1+\sqrt{1/\alpha})^2}{\Big(1-\sqrt{\beta/\alpha}-\sqrt{2\beta/\alpha\log(\e/\beta)}\Big)^2},
\end{equation*}
if $1-\sqrt{\beta/\alpha}-\sqrt{2\beta/\alpha\log(\e/\beta)}>0$, and otherwise $\gamma(\alpha,\beta) \rightarrow +\infty$.
Let also
\begin{equation}
\label{rho formula}
\rho(\alpha)\triangleq\max_{0\leq\beta\leq\alpha}{\frac{\beta}{2+2\gamma(\alpha,\beta)}}\cdot
\end{equation}
Then, $\rho(\alpha)>0$ for any $\alpha>0$. Moreover, we can guarantee that for almost every large system with ratio $n/m\rightarrow \alpha$, the asymptotic SL0 can recover the sparse  solutions satisfying $\norm{\sbb}{0}\leq \rho(\alpha)m$.
\end{coro}
\begin{proof}
To show $\rho(\alpha)>0$, simply note that
\begin{equation}
\lim_{\beta\rightarrow 0^+}{\frac{\beta}{2+2\gamma(\alpha,\beta)}}=0^+ . 
\end{equation}
For the second part, it suffice to apply Theorem \ref{theo: large gaussian case} with $n_0=\lceil \beta^* m \rceil $, where $\beta^*$ is the value of $\beta$ that maximizes $\gamma(\alpha,\beta)$ in (\ref{rho formula}).
\end{proof}

\section{Stability of the internal loop and its exponential convergence rate} \label{sec: stability}
From 
Fig.~\ref{fig: SL0 alg}, the steepest ascent steps in SL0 are of the form:
\begin{equation}
\label{eq: step}
\sbb_{i+1}=\sbb_i+\mu \sigma^2 \Db^T\Db\nabla F|_{\sbb_i}
\end{equation}
where $\Db^T\Db$ is the orthogonal projection on $\nl(\Ab)$ and $\mu$ is the step size parameter.  Until now, we have considered convergence of what we refer to as asymptotic version of SL0 (corresponding to $\mu\to 0$ and $L\to\infty$), in which the steps of the internal loop of Fig.~\ref{fig: SL0 alg} 
follow exactly along the steepest ascent trajectory . In this section, we study how to choose  the parameter $\mu$. For this part of the analysis, we assume the internal loop is repeated until convergence (corresponding to $L\to\infty$).



\begin{lemma}
Let $F=F_{\gamma',\sigma}$, where $\gamma'>\gamma=\gamma(n_0)$ and $\sigma>0$ is arbitrary. Let also $\lambda_{\min}$ and $\lambda_{\max}$ denote the smallest 
and largest eigenvalues of $-\Db \sigma^2 \Hb_F(\sbb) \Db^T$ respectively (note that the values of $\lambda_{\min}$ and $\lambda_{\max}$ depend on $\sbb$). Then for 
all $\sbb \in \Rm$
\begin{equation}
\lambda_{\max} \leq \frac{2}{1+\gamma}
\end{equation}
and for all $\sbb \in \Am$
\begin{equation}
\lambda_{\min} \geq \frac{2(\gamma'-\gamma)}{(1+\gamma)(\gamma'+\gamma'^2)}\cdot
\end{equation}
\end{lemma}

\begin{proof}
For convenience, let's define
\begin{equation}
\lambda'_{\max}\triangleq\frac{2}{1+\gamma}, \quad
\lambda'_{\min}\triangleq\frac{2(\gamma'-\gamma)}{(1+\gamma)(\gamma'+\gamma'^2)},
 \label{eq: lambda'}
\end{equation}
so that we need to show that
\begin{equation}
\label{eq: max min ineq}
\lambda_{\max}\leq\lambda'_{\max}, \quad \lambda_{\min}\geq\lambda'_{\min}.
\end{equation}

We know that for any matrix $\Mb$ with maximum and
minimum eigenvalues $\lambda_{\max}(\Mb)$ and $\lambda_{\min}(\Mb)$, $\Mb-\lambda \Ib$ is positive semi-definite if and only if $\lambda \leq \lambda_{\min}(\Mb)$. Moreover
$\Mb-\lambda \Ib$ is negative semi-definite if and only if $\lambda \geq \lambda_{\max}$.

To prove (\ref{eq: max min ineq}), we show that  $\Db (\sigma^2 \Hb_F(\sbb)+\lambda'_{\max}\Ib )\Db^T$ is positive semi-definite for all $\sbb \in \Rm$,
and $\Db (\sigma^2 \Hb_F(\sbb)+\lambda'_{\min}\Ib)\Db^T$ is negative semi-definite as long as $\sbb \in \Am$. Following steps of the proof of Lemma \ref{lemma: negdef}, the former follows from
\begin{equation}
\wb^T \Big(\sigma^2 \Hb_F(\sbb)+\lambda'_{\max}\Ib\Big) \wb \ge
(\lambda'_{\max}-\frac{2}{1+\gamma})\norm{\wb}{}^2 \ge 0.
\end{equation}
To show the second assertion, from (\ref{eq: gamma formula}) we obtain ${\norm{\wb_I}{}^2}/{\norm{\wb_I^c}{}^2}\leq \gamma$. 
Then, from (\ref{eq: square sum}) we have
\begin{equation}
\begin{split}
\wb^T \Big(\sigma^2 \Hb_F(\sbb)+\lambda'_{\min}\Ib\Big) \wb & \le
(\lambda'_{\min}-\frac{2}{1+\gamma})\norm{\wb_{I^c}}{}^2 \\ &+(\lambda'_{\min}+\frac{2}{\gamma^2+\gamma})\norm{\wb_{I}}{}^2
\le 0 \cdot
\end{split}
\end{equation}
Hence, $\Db (\sigma^2 \Hb_F(\sbb)+\lambda_{\min}'\Ib)\Db^T$ and $\Db (\sigma^2 \Hb_F(\sbb)+\lambda_{\max}'\Ib)\Db^T$ are negative and positive semi-definite respectively,
and (\ref{eq: max min ineq}) holds.

\end{proof}

\begin{theorem}\label{th: SL0 L to infty}
Let $F=F_{\gamma',\sigma}$, where $\gamma'>\gamma=\gamma(n_0)$.
Suppose also that:
\begin{equation}
\label{eq: stablity condition}
F(\sbb_i)\geq m - \frac{n_0}{2+2\gamma}\cdot
\end{equation}
Then, by setting
\begin{equation}
\label{eq: mu cond}
\mu=2/(\lambda'_{\min}+\lambda'_{\max}),
\end{equation}
where $\lambda'_{\max}$ and $\lambda'_{\min}$ are as defined in (\ref{eq: lambda'}), it is guaranteed that
\begin{equation}
\label{eq: CR cond}
\norm{\sbb_{i+1}-\sbb_{opt}}{}\leq \CR' \norm{\sbb_{i}-\sbb_{opt}}{},
\end{equation}
where $\sbb_{opt}$ is the maximizer of $F$ on $\Sm_\xb$,
$\sbb_{i+1}$ is as defined in (\ref{eq: step}),
and $\CR'\triangleq {(\lambda'_{\max}-\lambda'_{\min})}/{(\lambda'_{\max}+\lambda'_{\min})}$ determines the convergence rate. Moreover:
\begin{equation}
\label{eq: F increase}
F(\sbb_{i+1})\geq F(\sbb_i).
\end{equation}
\end{theorem}

\begin{proof}
The proof consists of the following steps.

{\bf Step 1}: From (\ref{eq: stablity condition}), $\sbb_i \in \Am$ and $\sbb_{opt} \in \Am$, where $\Am$ is as defined in Corollary~\ref{coro: convexity}. 
From Corollary \ref{coro: convexity}, $\Am$ is convex and $F$ is concave on $\Am$. Hence, $\sbb_{opt}$ satisfies:
\begin{equation}
\label{eq: optimal solution condition}
\sbb_{opt}=\sbb_{opt}+\mu \sigma^2 \Db^T\Db\nabla F|_{\sbb_{opt}} \Leftrightarrow \Db\nabla F|_{\sbb_{opt}}=0.
\end{equation}
Subtracting (\ref{eq: optimal solution condition}) from (\ref{eq: step}), we have
\begin{equation}
\sbb_{i+1}-\sbb_{opt}=\sbb_i-\sbb_{opt}+\mu \sigma^2 \Db^T \Db(\nabla F|_{\sbb_i}-\nabla F|_{\sbb_{opt}}).
\end{equation}
Multiplying by $\Db$ and setting $\Db\Db^T=\Ib$, we get
\begin{equation}
\Db(\sbb_{i+1}-\sbb_{opt})=\Db(\sbb_i-\sbb_{opt})+\mu\sigma^2 \Db(\nabla F|_{\sbb_{i}}-\nabla F|_{\sbb_{opt}}),
\end{equation}
From the mean value theorem, there exists a $t \in [0,1]$ such that $\sbb'\triangleq t\sbb_{opt}+(1-t)\sbb_i$ satisfies:
\begin{equation}
\Db(\sbb_{i+1}-\sbb_{opt})=\Db(\sbb_i-\sbb_{opt})+\mu\sigma^2 \Db \Hb_F(\sbb')(\sbb_i-\sbb_{opt}).
\end{equation}
Since $\{\sbb_i, \sbb_{opt}\} \in \Am$, it means that $\sbb' \in \Am$.
Also, since $(\sbb_i-\sbb_{opt}) \in \nl(\Ab)$, it is equal to its projection to $\nl(\Ab)$, that is, $\sbb_i-\sbb_{opt}=\Db^T\Db(\sbb_i-\sbb_{opt})$. Therefore, 
the above equation can be written as
\begin{equation}
\Db(\sbb_{i+1}-\sbb_{opt})=\Big(\Ib+\mu\sigma^2 \Db \Hb_F(\sbb') \Db^T\Big)\Db(\sbb_i-\sbb_{opt}).
\end{equation}
Since $(\sbb_{i}-\sbb_{opt})$ and $(\sbb_{i+1}-\sbb_{opt})$ are both in $\nl(\Ab)$, from (\ref{eq: norm formula}) we can write
\begin{equation}
\label{eq: norm cond}
\norm{\sbb_{i+1}-\sbb_{opt}}{} \leq \norm{\Ib+\mu\sigma^2 \Db \Hb_F(\sbb') \Db^T}{2}\cdot \norm{\sbb_i-\sbb_{opt}}{}.
\end{equation}

{\bf Step 2}: Let's define the Rate of Convergence (CR) as
\begin{equation}
\label{eq: CR ineq}
\begin{split}
& \CR  = \norm{\Ib+\mu\sigma^2 \Db \Hb_F(\sbb') \Db^T}{2} \\
& = \max\{|1- \mu \lambda_{\min}|,|1-\mu\lambda_{\max}|\}  \\
& = \max\{1- \mu \lambda_{\min},-1+\mu \lambda_{\min},-1+\mu\lambda_{\max},1-\mu\lambda_{\max}\}  \\
& = \max\{1- \mu \lambda_{\min},-1+\mu\lambda_{\max}\},
\end{split}
\end{equation}
where $\lambda_{\min}$ and $\lambda_{\max}$ are the smallest and largest eigenvalues of $-\Db \sigma^2 \Hb_F(\sbb') \Db^T$. The value of $\mu$ that
optimizes $\CR$ is $\mu=2/(\lambda_{\max}+\lambda_{\min})$, which results in
\begin{equation}
\CR=1-2\frac{\lambda_{\min}}{\lambda_{\max}+\lambda_{\min}}=\frac{\lambda_{\max}-\lambda_{\min}}{\lambda_{\max}+\lambda_{\min}}=\frac{\kappa-1}{\kappa+1},
\end{equation}
where $\kappa=\kappa(-\sigma^2 \Db\Hb_F(\sbb') \Db^T) = \lambda_{\max}/\lambda_{\min}$ denotes the condition
number of matrix $\Db$. With this definition, we have
\begin{equation}
\label{eq: CR formula}
\norm{\sbb_{i+1}-\sbb_{opt}}{}\leq \CR \norm{\sbb_{i}-\sbb_{opt}}{}.
\end{equation}
Computing  $\lambda_{\max}$ and $\lambda_{\min}$ in each step is not practical for large scale systems. Instead, we can find bounds on their values 
 using (\ref{eq: max min ineq}). Of course this bounds do not depend on $\sbb'$.

Choosing $\mu$ according to (\ref{eq: mu cond}) and considering (\ref{eq: CR ineq}), we have
\begin{equation}
\label{eq: final CR cond}
\norm{\Ib+\mu\sigma^2 \Db \Hb_F(\sbb') \Db^T}{2}\leq \frac{\lambda'_{\max}-\lambda'_{\min}}{\lambda'_{\max}+\lambda'_{\min}} = \CR'.
\end{equation}
Taking (\ref{eq: final CR cond}) together with (\ref{eq: norm cond}), we obtain (\ref{eq: CR cond}).

{\bf Step 3}: From the second order Taylor expansion of $F$ around $\sbb_i$ we have
\begin{equation}
\begin{split}
F(\sbb_{i+1})-F(\sbb_i)&=(\sbb_{i+1}-\sbb_i)^T \nabla F \\ & \quad + \frac{1}{2} (\sbb_{i+1}-\sbb_i)^T \Hb_{F} (\sbb_{i+1}-\sbb_i) 
\end{split}
\end{equation}
where $\nabla F = \nabla F |_{\sbb_i}$ and $\Hb_{F} = \Hb_{F} (\sbb'')$, for some point $\sbb''$ satisfying $\sbb'' = t \sbb_i + (1-t) \sbb_{i+1}$ 
for some $0 \leq t \leq 1$. Then, by substituting $\sbb_{i+1}-\sbb_i$ from (\ref{eq: step}) and factoring we get
\begin{multline}
\label{eq: taylor exp}
F(\sbb_{i+1})-F(\sbb_i)= \\
\frac{\mu^2\sigma^2}{2} \nabla F^T \Db^T \Big( \sigma^2 \Db \Hb_{F} \Db^T+(2/\mu)\Ib\Big)\Db\nabla F . 
\end{multline}
From (\ref{eq: mu cond}) and (\ref{eq: max min ineq}) we have
\begin{equation}
\label{eq: lambda max}
\lambda_{\max}\leq 2/\mu \cdot
\end{equation}
Now (\ref{eq: F increase}) is a straightforward conclusion of (\ref{eq: taylor exp}) and (\ref{eq: lambda max}).

\end{proof}
{\bf Remark 1. \ } The value of $\gamma<\gamma'<(n_0/2k)-1$ should be chosen carefully. If $\gamma'\rightarrow\gamma$, then $\lambda'_{\max}/\lambda'_{\min}\rightarrow\infty$
and $\CR' \rightarrow 1$. If $\gamma'\rightarrow (n_0/2k)-1$, then $c\rightarrow 1$ in (\ref{eq: c}), and the computational cost tends to infinity. In 
Section~\ref{sec: algorithm}, we discuss how to choose $\gamma'$ to have a reasonable convergence.

{\bf Remark 2. \ } Theorems~\ref{theo: guaranteed convergence} and~\ref{th: SL0 L to infty} prove convergence of SL0, provided that the internal loop is repeated 
until convergence is reached. The question remains to be answered is how to select the value of $L$  to guarantee that the internal loop is repeated until convergence is reached. This question is answered in Section~\ref{sec: algorithm}.

\section{ The noisy case} \label{sec: noisy}
Thus far we discussed the convergence and stability of SL0 in the noiseless case. Theorem~3 of~\cite{MohiBJ09} states that the maximizer of $F_\sigma$ is a good estimator 
of sparse solution even in the noisy case. In this section we investigate the choice of parameters that assure local concavity and, hence, convergence of SL0 
when data contains noise.

The following theorem is a modification of Theorem~3 of~\cite{MohiBJ09} and it provides conditions for convergence in noise.
\begin{theorem}
\label{theo: noisy case conv} Let $\Sm_\epsilon=\{\sbb  | \, \,
\norm{\Ab\sbb-\xb}{} < \epsilon\}$, where $\epsilon$ is an
arbitrary positive number, and assume that matrix $\Ab$ and
functions $\fs$ satisfy the conditions of Theorem \ref{theo: guaranteed convergence}.
Let $\sbb_0 \in \Sm_\epsilon$ be a sparse
solution. Assume the condition $k<n_0/(2+2\gamma)$, and
choose any $k'$ satisfying $k<k'<n_0/(2+2\gamma)$. We also choose the first term $\sigma_1$ and the scale factor $c$ according to
\begin{equation}
\label{eq: sigma 1 choice noisy}
\sigma_1=\frac{\norm{\hsb}{}}{\sqrt{k'(1+\gamma)}},
\end{equation}
\begin{equation}
\label{eq: c choice noisy}
c=\frac{2m}{2m+n_0/(2+2\gamma)-k'}<1
\end{equation}
and set $\sigma_j=\sigma_1 c^{j-1}$, $1 \leq j \leq J$, where $J$ is the index of the smallest term of the $\sigma$ sequence satisfying
\begin{equation}
\label{eq: choose sigma noisy}
\sigma_J \geq \frac{2\sqrt{m}\norm{\Ab}{2}\epsilon}{(1+\gamma)(k'-k)} > \sigma_{J+1} = c\sigma_J \cdot
\end{equation}
Then, following the steps of asymptotic SL0 and terminating at step $J$, one can achieve a solution within the distance $C\epsilon$ of the sparsest solution, where
\begin{equation}
\label{eq: error bound}
C=\Big(\frac{4m}{c(k'-k)\sqrt{\gamma+1}}+1\Big)\norm{\Ab}{2}.
\end{equation}
\end{theorem}

\begin{proof}
Let $\nb\triangleq\Ab\sbb_0-\xb$. Then, $\sbb_0 \in \Sm_\epsilon$
means that $\norm{\nb}{}<\epsilon$. Defining
$\tnb\triangleq\Ab^T\nb$, we have
\begin{displaymath}
\xb=\Ab\sbb_0+\nb=\Ab\sbb_0+\Ab\Ab^T\nb=\Ab\sbb_0+\Ab\tnb=\Ab(\sbb_0+\tnb)=\Ab\tsb,
\end{displaymath}
where $\tsb \triangleq \sbb_0+\tnb$. Let $\sbb_{\sigma}$ be the
maximizer of $\Fs$ on
$\oureq$, as defined in Theorem~1 of~\cite{MohiBJ09}. Note that, $\sbb_{\sigma}$ is not necessarily
the maximizer of $\Fs$ on the whole $\Sm_\epsilon$. The argument is similar to that of Theorem~3 in \cite{MohiBJ09}. From
(\ref{eq: deviation}) in lemma~\ref{lemma: Derivation} and (\ref{eq: choose sigma noisy}), we have
\begin{multline}
    \norm{\tsb-\sbb_0}{}=\norm{\tnb}{}<\norm{\Ab}{2}\epsilon \Rightarrow \\ |F_{\sigma_j}(\tsb)-F_{\sigma_j}(\sbb_0)|\leq 
\frac{2\sqrt{m}}{\sigma_j(1+\gamma)}\norm{\tsb-\sbb_0}{} \leq k'-k \cdot
\end{multline}
Hence
\begin{equation}
\label{eq: s tilde formula}
F_{\sigma_j}(\sbb_0)\geq m-k \Rightarrow F_{\sigma_j}(\tsb)\geq m-k' \cdot
\end{equation}
The vector $\sbb_0$ does not necessarily satisfy $\oureq$, however,
we have chosen $\tsb$ to be the projection of $\sbb_0$ onto the
subspace $\oureq$. Hence, $\tsb$ satisfies $\oureq$. Moreover, $\Fs(\tsb)>m-k'>m-n_0/(2+2\gamma)$, hence $\tsb\in\Am$, and by optimizing $\Fs$ from an arbitrary point 
in $\Am$ we are guaranteed a solution $\sbb_*$ for which $\Fs(\sbb_*) \geq m-k'$. Now, using Lemma \ref{lemma: choose sigma 1 and c}, it is easy to conclude that for $\sigma_1$ and $c$ chosen according to (\ref{eq: sigma 1 choice noisy}) and (\ref{eq: c choice noisy}),
\begin{equation}
F_{\sigma_1}(\hsb) \geq m-k'
\end{equation}
and
\begin{equation}
F_{\sigma}(\sbb) \geq m-k' \Rightarrow F_{c\sigma}(\sbb) \geq m-n_0/(2+2\gamma).
\end{equation}
Following the steps of the proof  of Theorem \ref{theo: guaranteed convergence}, but with the sparsity factor $k$ replaced by
 $k'$, we can conclude that
\begin{equation}
\label{eq: final sigma formula}
F_{\sigma_J}(\sbb_{J})\geq m-k' \cdot
\end{equation}
Using Lemma~\ref{lemma: L0 solution approximation}, (\ref{eq: s tilde formula}) and (\ref{eq: final sigma formula}), we then have
\begin{equation}
\norm{\sbb_{J}-\tsb}{}\leq 2\sqrt{m(\gamma+1)}\sigma_J \leq \frac{4m\norm{\Ab}{2}\epsilon}{c\sqrt{1+\gamma}(k'-k)}
\end{equation}
and
\begin{equation}
\begin{split}
\norm{\sbb_{J}-\sbb_0}{} & \leq \norm{\sbb_{J}-\tsb}{} + \norm{\tsb-\sbb_0}{} \\ &\leq \frac{4m\norm{\Ab}{2}\epsilon}{c\sqrt{1+\gamma}(k'-k)} + \norm{\Ab}{2}\epsilon = C\epsilon \cdot
\end{split}
\end{equation}
\end{proof}

{\bf Remark 1. \ } If $k'\rightarrow k$, the error bound tends to infinity in (\ref{eq: error bound}). If $k'\rightarrow n_0/(2+2\gamma)$, the
computational cost would tend to infinity as $c\to 1$ in (\ref{eq: c choice noisy}). Hence $k'$ should be chosen suitably between the two values. A simple sub-optimal choice is presented in the next section.

{\bf Remark 2. \ } In Theorem 3 of \cite{MohiBJ09}, we proved that by suitably choosing $\sigma$ proportional to the noise level, we can bound the 
Euclidean distance between the maximizer of $F_{\sigma}$ and the sparse solution by order of the noise standard deviation. Experiment 2 of~\cite{MohiBJ09} 
(Section~IV, Fig.~4) confirmed the result of Theorem 3 of \cite{MohiBJ09}. Here, (\ref{eq: choose sigma noisy}) and (\ref{eq: error bound}) also confirm this result. As can be seen from (\ref{eq: error bound}), the estimation error depends linearly on the system noise.

\section{Finalizing the convergence analysis} \label{sec: algorithm}

At this point we have acquired all the tools necessary for ensuring the convergence of the external loop, stability of the steepest ascent (internal loop), and robustness 
against noise for SL0. The only parameter we have not yet discussed is $L$ (the number of iterations of the internal loop shown in Fig.~\ref{fig: SL0 alg}). In this section, 
we put all the previous results together
and provide values for all the parameters that are sufficient to guarantee successful convergence of SL0.

We present results for three cases. In the first case, we assume that suitable values of $n_0$ and $\gamma = \gamma(n_0)$ are known, such that $\norm{\sbb_0}{0}<n_0/(2+2\gamma)$. 
In this case, the values of the parameters that guarantee the convergence are summarized in Fig.~\ref{fig: the alg} and the convergence is proved in Theorem~\ref{theo: ultimate algorithm known}.

In the second case, $\gamma$ is assumed unknown and we consider a large Gaussian  matrix $\Ab$, and use the almost sure results  of Section~\ref{sec: large random gaussian} to determine 
$n_0$ and $\gamma$. 
The values of the parameters for this case that guarantee convergence are summarized in Fig.~\ref{fig: the init}.
For a random matrix $\Ab$ with i.i.d~and zero-mean Gaussian entries, Theorem~\ref{theo: ultimate algorithm unknown} shows that using these parameters the sparse solution of 
$\Ab \sbb = \xb$ can be found with probability approaching 1 as the size of the system grows, as long as $\norm{\sbb_0}{0}<\rho(\alpha)m$.
Moreover, it is shown that the complexity of SL0 grows as $m^2$, which is faster than the state of the art $m^{3.5}$ associated with Basic Pursuit and is comparable with Matching Pursuit.

The third case deals with multiple source recovery where the sparsest solutions of multiple USLE's with the same coefficient matrix are recovered at once.  Multiple source recovery may be viewed in the context of SCA \cite{GribL06} for Blind Source Separation. ~ In Experiment~6 of~\cite{MohiBJ09} we observed that implementing SL0 for multiple source recovery in matrix multiplication form can make it faster than the SL0 algorithm for single solution recovery. Theorem~\ref{theo: BSS case} shows that this approach can speed up SL0 to the 
order of $m^{1.376}$.

\subsection{Case of known $\gamma$} \label{sec: known gamma}
Putting the results of previous sections together, the following theorem shows that if the values of the parameters are chosen as summarized in Fig.~\ref{fig: the alg}, then SL0 will converge to the sparsest solution. The proposed value for $L$ can be seen in the step~\ref{step: L def} of the figure. Note also that the notation of Fig.~\ref{fig: SL0 alg} has been changed slightly in Fig.~\ref{fig: the alg} to match the convergence proof given next.


\begin{theorem}[The case of known $n_0$ and $\gamma$]
\label{theo: ultimate algorithm known}
Let $\gamma=\gamma(n_0)$ and, without loss of generality, assume matrix $\Ab$ has orthonormal rows. Let $\xb=\Ab\sbb_0+\nb$ for some $\norm{\nb}{}\leq \epsilon$ and $\norm{\sbb_0}{0}\leq k<n_0/2(1+\gamma)$. Let $\Delta\triangleq\frac{n_0/2(1+\gamma)-k}{4m}$. Then the algorithm given in Fig.~\ref{fig: the alg} can recover $\sbb_0$ within a distance $\delta>C\epsilon$, where
\begin{equation}
\label{eq: strict C}
C\triangleq\Big(\frac{4}{\Delta\sqrt{\gamma+1}}+1\Big)\norm{\Ab}{2} \cdot
\end{equation}

\begin{figure}
  \centering \vrule%
  \ifonecol
    \begin{minipage}{15.7cm} 
  \else
    \begin{minipage}{9cm} 
  \fi
\hrule \vspace{0.5em}
\hspace*{-1.1em}
  \ifonecol
    \begin{minipage}{9.5cm} 
  \else
    \begin{minipage}{8.5cm} 
  \fi
      {
        \footnotesize

        \begin{itemize}
        \item Initialization:

           \begin{enumerate}
             \item $\Delta\leftarrow\frac{n_0/2(1+\gamma)-k}{4m}$
             \item $k'\leftarrow k+m\Delta$
             \item $k''\leftarrow k+2m\Delta$
             \item $\frac{n_0}{2(1+\gamma')}\leftarrow k+3m\Delta$ (\ie $\gamma'\leftarrow \frac{n_0}{2(k+3m\Delta)}-1)$
             \item $F \leftarrow F_{\gamma'}$
	         \item $\delta' \leftarrow \delta - \norm{\Ab}{2}\epsilon$
             \item $\sigma_1 \leftarrow \norm{\Ab^T\xb}{}/\sqrt{n_0/(2+2\gamma')}$
             \item $\sigma_J \leftarrow \delta'/2\sqrt{m(\gamma'+1)}$
             \item $J\leftarrow \lceil\frac{\log(\sigma_1)-\log(\sigma_J)}{\log(1+\Delta/2)}\rceil+1$
             \item $\log{(c)}\leftarrow -\frac{\log(\sigma_1)-\log(\sigma_J)}{J-1}$
             \item $\sigma_j\leftarrow \sigma_1 c^{j-1} (1\leq j \leq J)$
             \item $\lambda'_{\max}\leftarrow \frac{2}{1+\gamma}$
             \item $\lambda'_{\min}=\frac{2(\gamma'-\gamma)}{(1+\gamma)(\gamma'^2+\gamma')}$
             \item $\mu\leftarrow 2/(\lambda'_{\min}+\lambda'_{\max})$
             \item $\kappa' \leftarrow \lambda'_{\max}/\lambda'_{\min}$
             \item $\CR'\leftarrow \frac{\kappa'-1}{\kappa'+1}$
             \item $L\leftarrow \lceil \frac{-\log(\Delta/4)-1/2 \log(\gamma'+1)}{-\log(\CR')}\rceil+1$ \label{step: L def}
           \end{enumerate}

        \item For $j=1,\dots,J$:

          \begin{enumerate}
             \item \hspace{-0.8em} $\sigma\leftarrow\sigma_j$.
             \item \hspace{-0.8em} If $j\geq 2$, $\sbb_{j,1}\leftarrow\sbb_{j-1,L}$. If $j=1$, $\sbb_{1,1}\leftarrow\Ab^T\xb$
             \item \hspace{-0.8em} For $l=1,\dots,L-1$:
             \begin{itemize}
                \item $\sbb_{j,l+1}\leftarrow\sbb_{j,l}+\mu\sigma^2\Db^T\Db\nabla F_{\sigma}|_{\sbb_{j,l}}$
             \end{itemize}
          \end{enumerate}

        \item Output is $\sbb_{out}\leftarrow\sbb_{J,L}$.
        \end{itemize}
      }
    \end{minipage}
    \vspace{1em} \hrule
  \end{minipage}\vrule \\
\caption{The SL0 algorithm for the case of known $n_0$ and $\gamma(n_0)$ and $\Ab$ with orthonormal rows , with parameters shown that guarantee convergence to the sparsest solution. $\sbb_{j,l}$ is the solution estimate at the corresponding iteration.} \label{fig: the alg}
\end{figure}

\end{theorem}

\begin{proof}
The proof is constructed using the following steps:

{\bf Step 1}: Let's set $\tsb=\sbb_0+\Ab^T\nb$, then we have $\tsb\in\Sm_{\xb}$ and also $F_{\sigma}(\tsb)\geq m-k'$ for any $\sigma\geq\sigma_J$. Assume that we have
\begin{equation}
\label{eq: sigma ineq}
\sigma\geq\sigma_J=\frac{\delta'}{2\sqrt{m(\gamma'+1)}} \cdot
\end{equation}
Then, from Lemma~\ref{lemma: Derivation} we have 
\begin{equation}
\label{eq: F ineq}
|F_{\sigma}(\sbb_0)-F_{\sigma}(\tsb)|\leq \frac{2\sqrt{m}}{(1+\gamma')\sigma}\norm{\sbb_0-\tsb}{} \cdot
\end{equation}
Since $\norm{\sbb_0-\tsb}{}\leq \norm{\Ab^T}{2}\cdot\norm{\nb}{}\leq \norm{\Ab}{2}\epsilon$ and
\begin{equation}
\label{eq: delta ineq}
\delta'\triangleq\delta-\norm{\Ab}{2}\epsilon \geq \frac{4}{\sqrt{\gamma+1}\Delta}\norm{\Ab}{2}\epsilon,
\end{equation}
from (\ref{eq: sigma ineq}), (\ref{eq: F ineq}) and (\ref{eq: delta ineq}) we have
\begin{equation}
\label{eq: F final ineq}
|F_{\sigma}(\sbb_0)-F_{\sigma}(\tsb)| \leq m\Delta \Rightarrow F_{\sigma}(\tsb) \geq F_{\sigma}(\sbb_0)-m\Delta=m-k' .
\end{equation}

{\bf Step 2}: We show that for any $1 \leq j \leq J$, if $F_{\sigma_j}(\sbb_{j,1})\geq m-\frac{n_0}{2+2\gamma'}$, then $F_{\sigma_j}(\sbb_{j,L})\geq m-k''$, where the notations $\sbb_{j,i}$ and $k''$ are defined in Fig.~\ref{fig: the alg}.
Let $\sbb_{opt}$ be the maximizer of $F_{\sigma_j}$ on $\Sm_\xb$. Hence, $F_{\sigma_j}(\sbb_{opt}) \geq F_{\sigma_j}(\sbb_{j,1}) \geq m-\frac{n_0}{2+2\gamma'}$ and from
Lemma \ref{lemma: L0 solution approximation},
\begin{equation}
\norm{\sbb_{j,1}-\sbb_{opt}}{} \leq 2\sqrt{m(\gamma+1)}\sigma.
\end{equation}
From (\ref{eq: CR cond}), we conclude
\begin{equation}
\begin{split}
\norm{\sbb_{j,L}-\sbb_{opt}}{} & \leq (\CR')^{L-1} \norm{\sbb_{j,1}-\sbb_{opt}}{} \\ & \leq \frac{\Delta\sqrt{\gamma'+1}}{4}\Big(2\sqrt{m(\gamma'+1)}\sigma_j\Big) \\ &\leq \frac{\sqrt{m}\Delta\sigma_j (1+\gamma')}{2},
\end{split}
\end{equation}
where the second inequality holds when the value of $L$ is defined as in the Step~\ref{step: L def} of Fig.~\ref{fig: the alg}.
Hence, from Lemma~\ref{lemma: Derivation}
\begin{equation}
\label{eq: F delta}
|F_{\sigma_j}(\sbb_{opt})-F_{\sigma_j}(\sbb_{j,L})|\leq \frac{2\sqrt{m}}{\sigma_j (1+\gamma')} \norm{\sbb_{j,L}-\sbb_{opt}}{} = m\Delta.
\end{equation}
Therefore, from (\ref{eq: F final ineq}) and (\ref{eq: F delta}) we have
\begin{equation}
F_{\sigma_j}(\sbb_{j,L}) \geq F_{\sigma_j}(\sbb_{opt}) - m\Delta \geq F_{\sigma_j}(\tsb) - m\Delta \geq m-k''.
\end{equation}
{\bf Step 3}: We show that if $F_{\sigma_{j-1}}(\sbb_{j-1,L})\geq m-k''$, then
\begin{equation}
F_{\sigma_j}(\sbb_{j,1})\geq m- n_0/(2+2\gamma') \cdot
\end{equation}
From the algorithm of Fig.~\ref{fig: the alg}, we know that
\begin{equation}
c\geq\frac{1}{1+\Delta/2}=\frac{2m}{2m+m\Delta} \cdot 
\end{equation}
Then, choosing $A=k''=k+2m\Delta$ and $B=n_0/(2+2\gamma')=k + 3m\Delta$ in Lemma \ref{lemma: choose sigma 1 and c} and substituting $\sbb_{j,1}=\sbb_{j-1,L}$, we have
\begin{equation}
F_{\sigma_{j-1}}(\sbb_{j-1,L})\geq m-k'' \Rightarrow F_{\sigma_j}(\sbb_{j,1})\geq m-n_0/(2+2\gamma') \cdot
\end{equation}
{\bf Step 4}: Here, we prove by induction on $j$ that $F_{\sigma_j}(\sbb_{j,L})\geq m-k''$. In the first step, we have $\sbb_{1,1}=\Ab^T\xb$ and
\begin{equation}
\sigma_1 = \norm{\Ab^T\xb}{}/\sqrt{n_0/(2+2\gamma')} \cdot
\end{equation}
Hence, from Lemma \ref{lemma: choose sigma 1 and c}
\begin{equation}
F_{\sigma_1}(\sbb_{1,1})\geq m- n_0/(2+2\gamma'),
\end{equation}
and from Step 2
\begin{equation}
F_{\sigma_1}(\sbb_{1,L})\geq m- k'' \cdot
\end{equation}
Assume that
\begin{equation}
F_{\sigma_{j-1}}(\sbb_{j-1,L})\geq m- k''
\end{equation}
for some $j$. Then from the results of Step 3 and noting from Fig~\ref{fig: the alg} that $\sbb_{j,1}=\sbb_{j-1,L}$, we obtain
\begin{equation}
F_{\sigma_j}(\sbb_{j,1})=F_{\sigma_j}(\sbb_{j-1,L}) \geq m- n_0/(2+2\gamma')
\end{equation}
and from Step 2,
\begin{equation}
F_{\sigma_j}(\sbb_{j,L})\geq m- k''.
\end{equation}
We can conclude then that
\begin{equation}
\label{F out}
F_{\sigma_J}(\sbb_{out}) = F_{\sigma_J}(\sbb_{J,L})\geq m- k''\geq m-n_0/(2+2\gamma').
\end{equation}
{\bf Step 5}: From Lemma \ref{lemma: L0 solution approximation}, (\ref{F out}), (\ref{eq: F final ineq}) and the choice of $\sigma_J$ given in step 8 of Fig.~\ref{fig: the alg}, we have
\begin{equation}
\norm{\hsb-\sbb_{out}}{}\leq 2\sqrt{m(\gamma'+1)}\sigma_J=\delta'
\end{equation}
and
\begin{equation}
\norm{\sbb_0-\sbb_{out}}{} \leq \norm{\hsb-\sbb_{out}}{} + \norm{\sbb_0-\hsb}{}\leq \delta' + \norm{\Ab}{2}\epsilon = \delta \cdot
\end{equation}
This completes the proof of convergence of SL0.
\end{proof}

{\bf Remark 1. \ } In noiseless case ($\epsilon = 0$), SL0 can recover the sparsest solution within a distance $\delta$, for some $\delta>0$, in a finite number of  steps. 
But as $\delta \rightarrow 0$, $\sigma_J$, i.e. the last value of $\sigma$, tends to zero according to step 8 of Fig.~\ref{fig: the alg}, and $J$ tends to $\infty$ according to step 9. Hence, the complexity of the algorithm tends to infinity.

{\bf Remark 2. \ } Note that the algorithm does not require the exact value of the $\ell^0$ norm. Only an upper bound $k$ is necessary.
\subsection{Case of unknown $\gamma$} \label{sec: unknown gamma}
For a large Gaussian $\Ab$, we can use the a.s.~results of Section~\ref{sec: large random gaussian} to find $n_0$ and $\gamma(n_0)$, and thus obtain the initialization of SL0 shown in
Fig.~\ref{fig: the init}. The following theorem guarantees convergence of the algorithm in Fig.~\ref{fig: the init}.

\begin{theorem}[the case of unknown $n_0$ and $\gamma$]
\label{theo: ultimate algorithm unknown}
Let $\Ab$ be an $n \times m$ Gaussian matrix, and $n/m\rightarrow \alpha>0$ as $m \rightarrow \infty$. Let’s fix $r<\rho(\alpha)$ and let 
$\Prob_m$ denote the probability that the algorithm in Fig.~\ref{fig: the init} can recover any $\sbb_0$ from $\xb=\Ab\sbb_0+\nb$ within Euclidean distance of $\delta=C'\epsilon$, as long as $\norm{\sbb_0}{0}<rm$, $\norm{\sbb_0}{}\leq 1$, and $\norm{\nb}{}<\epsilon$, where
\begin{equation}
\label{eq: soft C}
C'\triangleq\bigg(\frac{16}{\Big(\rho(\alpha)-r\Big)\sqrt{\gamma+1}}+1\bigg)(1+\sqrt{\alpha}).
\end{equation}

\begin{figure}
  \centering \vrule%
  \ifonecol
    \begin{minipage}{15.7cm} 
  \else
    \begin{minipage}{9cm} 
  \fi
\hrule \vspace{0.5em}
\hspace*{-1.1em}
  \ifonecol
    \begin{minipage}{9.5cm} 
  \else
    \begin{minipage}{8.5cm} 
  \fi
      {
        \footnotesize

        \begin{itemize}
        \item Initialization:
            \begin{enumerate}
                \item $\beta^* \leftarrow$ maximizer of $\beta/(2+2\gamma(\alpha,\beta))$ on $0 \leq \beta \leq \alpha$
                \item $\gamma \leftarrow \gamma(\alpha,\beta^*)$
                \item $n_0 \leftarrow \lceil \beta^* m \rceil$
                \item $k\leftarrow \lceil r m \rceil$
                \item $\delta \leftarrow C'\epsilon$, where $C'$ is defined in (\ref{eq: soft C})
                \item $\sigma_1 \leftarrow (1+\sqrt{\alpha})(1+\sqrt{\alpha}+\epsilon)$. (This step replaces step~7 of Fig.~\ref{fig: the alg}).
		\item Do initialization steps $1 \cdots 6$ and $8 \cdots 17$ of Fig.~\ref{fig: the alg}.
             \end{enumerate}
        \end{itemize}
      }
    \end{minipage}
    \vspace{1em} \hrule
  \end{minipage}\vrule \\
\caption{SL0 initialization parameters for the case of unknown $\gamma$. Step~6 here replaces step~7 of Fig.~\ref{fig: the alg}.} \label{fig: the init}
\end{figure}

\noindent Then, we have $\Prob_m \to 1$ as $m \to \infty$. Moreover, the complexity of the algorithm is $O(m^2)$.
\end{theorem}

\begin{proof}
We know from Theorem~\ref{theo: large gaussian case} that $\Prob\la \gamma(n_0)>\gamma \ra \to 0$ as $m \to \infty$. Moreover, $\Prob \la \norm{\Ab}{2}>\sqrt{\alpha}+1 \ra \to 1$  
as $m \to \infty$~\cite{DaviS01, Ledo01}. Therefore noting $\xb = \Ab\sbb_0+\nb$ we have
\begin{equation}
\label{eq: sigma 1 m ineq}
\Prob \la \norm{\Ab^T\xb}{} \leq (1+\sqrt{\alpha})^2+(1+\sqrt{\alpha})\epsilon \ra \to 1
\end{equation}
as $m \to \infty$, because $\norm{\sbb_0}{2}\leq 1$ and $\norm{\nb}{}<\epsilon$. This means that the condition imposed by step~6 of Fig.~\ref{fig: the init} is stricter than that imposed by 
step~7 of Fig.~\ref{fig: the alg}. Thus, all the conditions of Theorem~\ref{theo: ultimate algorithm known} also apply for the algorithm in Fig.~\ref{fig: the init}. Hence, the Euclidean distance between the final solution and the sparsest 
solution is less than $C\epsilon$, \ie
\begin{equation}
\norm{\sbb_{out} - \sbb_{0}}{2} < C\epsilon
\end{equation}
where $C$ is as defined in (\ref{eq: strict C}). Moreover, $\Prob \la C<C'\ra \to 1$ as $m \to \infty$, where $C'$ is as defined in (\ref{eq: soft C}). 
Hence, the accuracy is better than $C'\epsilon$ with probability tending to 1, which completes the proof of the convergence result.

From Fig.~\ref{fig: the alg}, it is clear that the computational complexity of SL0 is $O(mnJL)$ and since $n/m\rightarrow \alpha>0$, we can assume $n=O(m)$.
To obtain the final complexity result, we show that $J=O(1)$ and $L=O(1)$ as $m\rightarrow\infty$. According to Fig.~\ref{fig: the alg},
\begin{equation}
J<\frac{\log(\sigma_1/\sigma_J)}{\log(1+\Delta/2)}+2.
\end{equation}
From the initialization of $\Delta$ shown in Fig.~\ref{fig: the alg},
\begin{equation}
\label{eq: Delta Limit}
\lim_{m\rightarrow\infty} \Delta>\frac{\beta^*/(2+2\gamma)-r}{4}=\frac{\rho(\alpha)-r}{4}>0
\end{equation}
and
\begin{equation}
\lim_{m\rightarrow\infty} \log(1+\Delta/2)>0 \cdot
\end{equation}
Hence to show that $J=O(1)$, we need to show that
\begin{equation}
\label{eq: sigma 1 ineq}
\lim_{m\rightarrow\infty} \sqrt{m}\sigma_1<\infty
\end{equation}
and
\begin{equation}
\label{eq: sigma J ineq}
\lim_{m\rightarrow\infty} \sqrt{m}\sigma_J>0 \cdot
\end{equation}
To show (\ref{eq: sigma 1 ineq}) note that
\begin{multline}
\label{eq: m n0 ineq}
\lim_{m\rightarrow\infty} \frac{\sqrt{m}}{\sqrt{n_0/(2+2\gamma')}}=\lim_{m\rightarrow\infty}\frac{1}{\sqrt{(k+3m\Delta)/m}}\\ \leq\frac{2}{\sqrt{\rho(\alpha)+3r}}<\frac{1}{\sqrt{r}}\cdot
\end{multline}
With $\sigma_1$ given in Fig.~\ref{fig: the init}, (\ref{eq: sigma 1 ineq}) becomes an obvious conclusion of (\ref{eq: sigma 1 m ineq}) and (\ref{eq: m n0 ineq}).
To show (\ref{eq: sigma J ineq}), note that from Fig.~\ref{fig: the alg}, we obtain
\begin{multline}
\label{eq: sigma J m ineq}
\lim_{m\rightarrow\infty} \sqrt{m}\sigma_J=(\delta'/2)\lim_{m\rightarrow\infty}\sqrt{1/(1+\gamma')}\\ >(\delta'/2)\sqrt{3/4(1+\gamma)}>0,
\end{multline}
where we have used the fact that
\begin{equation}
\label{eq: gamma' gamma relation}
n_0/(1+\gamma')=3 n_0/4(1+\gamma)+k/4\Rightarrow 1/(1+\gamma')>3/4(1+\gamma) .
\end{equation}
Next, we show that $L=O(1)$. Note that from Fig.~\ref{fig: the alg}
\begin{equation}
L<\frac{-\log(\Delta/4)}{-\log(\CR')}+2.
\end{equation}
From (\ref{eq: Delta Limit}) we know that $-\log(\Delta/4)$ is bounded. Hence, to complete the proof of $L=O(1)$, we need to show that
\begin{equation}
\lim_{m\rightarrow\infty} \log(\CR')<0 \Leftrightarrow \lim_{m\rightarrow\infty} \CR'<1 .
\end{equation}
From the definition of $\lambda'_{\min}$, $\lambda'_{\max}$, and $\kappa'$ in Fig.~\ref{fig: the alg},
\begin{equation}
\kappa'=(\gamma'^2+\gamma')/(\gamma'-\gamma).
\end{equation}
Observe that
\begin{equation}
\begin{split}
\gamma'-\gamma&=\frac{n_0}{2(k+3m\Delta)}-\frac{n_0}{2(k+4m\Delta)}\\
& =\frac{n_0 m \Delta}{2(k+3m\Delta)(k+4m\Delta)}
\end{split}
\end{equation}
and
\begin{equation}
\label{eq: dif gamma' gamma bound}
\begin{split}
\lim_{m\rightarrow\infty} \frac{\gamma'-\gamma}{1+\gamma'}&= \lim_{m\rightarrow\infty} \frac{m\Delta}{k+4m\Delta}=\lim_{m\rightarrow\infty} \frac{\Delta}{k/m+4\Delta} \\ &>\frac{\Delta}{4\Delta+r}=\frac{\Delta}{\rho(\alpha)}>0.
\end{split}
\end{equation}
Also note that from (\ref{eq: gamma' gamma relation}), we have
\begin{equation}
\label{eq: gamma' gamma relation new}
\gamma'<4/3(1+\gamma)-1.
\end{equation}
Then, from (\ref{eq: dif gamma' gamma bound}) and (\ref{eq: gamma' gamma relation new}) one can conclude
\begin{equation}
\lim_{m\rightarrow\infty} \kappa'<\infty
\end{equation}
and
\begin{equation}
\lim_{m\rightarrow\infty} \CR'=1-2\lim_{m\rightarrow\infty} \frac{1}{\kappa'+1}<1.
\end{equation}
\end{proof}

{\bf Remark 1. \ } In \cite{MohiBJ09}, we experimentally observed that the optimal value of $L$ is a small constant 
(Fig. 5, Experiment 2, section IV). Here, we proved that $L$ is bounded as $m \to \infty$.

\subsection{Multiple Sparse Solution Recovery Case} \label{sec: BSS}

Thus far, we discussed the recovery of the sparsest solution of USLE containing a single measurement vector. In SCA applications one deals with multiple measurement vectors.  

The resulting system of equations can be written in matrix form:
\begin{equation}
\label{BSS model}
\Xb=\Ab\Sb+\Nb,
\end{equation}
where $\Xb\triangleq [\xb(1),\dots,\xb(T)] \in \Rr^{n \times T}$, $\Sb\triangleq [\sbb(1),\dots,\sbb(T)] \in \Rr^{m \times T} $ and $\Nb\triangleq [\nb(1),\dots,\nb(T)] \in \Rr^{n \times T}$. As observed in Experiment~6 of~\cite{MohiBJ09}, when we
apply the MSL0 (SL0 for multiple sparse recovery) of Fig.~\ref{fig: the BSS case}, the overall computational complexity reduces as compared to $T$ separate applications of the vector version of SL0. The following theorem  supports this observation.
\begin{theorem}
\label{theo: BSS case}
Under the conditions of Theorem \ref{theo: ultimate algorithm unknown},  using the algorithm shown in Fig.~\ref{fig: the BSS case} to recover the sparsest solutions satisfying  (\ref{BSS model})
reduces average computational complexity of each individual solution $\xb(t), 1 \leq t \leq T$ to $O(m^{1.376})$ as $T/m \to \infty$.
\end{theorem}
\begin{proof}
Note that the only computationally expensive part of the algorithm is step 3 of the loop in Fig. \ref{fig: the BSS case}, where we multiply the $m\times m$ matrix $\Db^T\Db$ by the
$m\times 1$ vector $\nabla F_{\sigma}|_{\sbb_{j,l}}$, and also the initialization of $\sbb_{1,1}$, where we compute $\Ab^T\xb$. This is because
these two steps are of order $O(m^2)$, and all the other computations are at most of order $O(m)$.
Analogous to approach of Experiment 6 in \cite{MohiBJ09}, we use the matrix form (\ref{BSS model}). We replace the final loop
with steps shown in Fig. \ref{fig: the BSS case}, and perform $m \times T$ matrix multiplication using $\lceil T/m \rceil$ multiplications of $m \times m$ matrices
using Coppersmith-Winograd algorithm~\cite{CoppW90}.
The overall complexity is $T/m$ times $O(m^{2.376})$, or equivalently, $T$ times $O(m^{1.376})$, meaning that per sample complexity is $O(m^{1.376})$.

\begin{figure}
  \centering \vrule%
  \ifonecol
    \begin{minipage}{15.7cm} 
  \else
    \begin{minipage}{9cm} 
  \fi
\hrule \vspace{0.5em}
\hspace*{-1.1em}
  \ifonecol
    \begin{minipage}{9.5cm} 
  \else
    \begin{minipage}{8.5cm} 
  \fi
      {
        \footnotesize
        \def\baselinestretch{1}

        \begin{itemize}
	\item Initialization: repeat initialization steps $1 \cdots 17$ of Fig.~\ref{fig: the alg}
        \item For $j=1,\dots,J$:
          \begin{enumerate}
             \item \hspace{-0.8em} $\sigma\leftarrow\sigma_j$.
             \item \hspace{-0.8em} If $j\geq 2$, $\Sb_{j,1}\leftarrow\Sb_{j-1,L}$. If $j=1$, $\Sb_{1,1}\leftarrow\Ab^T\Xb$
             \item \hspace{-0.8em} For $l=1,\dots,L-1$:
             \begin{itemize}
                \item $\Sb_{j,l+1}\leftarrow\Sb_{j,l}+\mu\sigma^2\Db^T\Db\nabla F_{\sigma}|_{\Sb_{j,l}}$
             \end{itemize}
          \end{enumerate}

        \item Output is $\Sb_{out}\leftarrow\Sb_{J,L}$.
        \end{itemize}
      }
    \end{minipage}
    \vspace{1em} \hrule
  \end{minipage}\vrule \\
\caption{MSL0 (SL0 for multiple sparse recovery). $\Sb_{j,i}$ is our estimation of the matrix of sparse solutions at the corresponding level. } \label{fig: the BSS case}
\end{figure}
\end{proof}

\section{Conclusion}

We had recently proposed the SL0 algorithm, which we showed empirically to be efficient and accurate for recovery of sparse solutions using $\ell^0$ minimization ~\cite{MohiBJ09}. Its convergence properties, however, were only partially analyzed, so the theoretical justification for SL0 remained incomplete. The current paper provides the theoretical justification for SL0.   

Several results were presented.  First, general results were derived showing that a judicial choice of parameter values guarantees that SL0 converges to the sparsest 
solution provided that the given system satisfies the recovery conditions.  These conditions were derived in 
terms of the lower asymmetric RIC and Eucleadian norm of the system. We then adapted the convergence results for the special case where the system is a large Gaussian matrix.
Next, we showed that convergence of SL0 can be similarly guaranteed in the case of noise.  The noise 
results combined with our previous work and numerical experiments presented in ~\cite{MohiBJ09} indicate that SL0 exhibits good robustness properties in noise. Lastly, 
we provided the complete parameter setting of SL0, that guaranteed recovery of the sparsest solutions in the case of general as well as Gaussian system.  
We then extended the SL0 algorithm to the case of multiple measurement vectors and provided the necessary parameter settings for the convergence. 

Also presented were computational complexity results for SL0 in the cases of single and multiple measurement vectors. We showed that in the limiting case $m\rightarrow\infty$ and $n/m\rightarrow \alpha>0$, the complexity is $O(m^2)$ and is comparable to that of orthogonal  MP  techniques.  Further,   we showed that recovering multiple sparse solutions simultaneously by using  MSL0   reduces complexity per individual solution to $O(m^{1.376})$.

The main purpose of the presented results is to fulfill the need for theoretical justification of SL0.  A number of papers have stated that RIP provides a strict condition for analysis of sparse recovery algorithms and it typically leads to unnecessarily pessimistic choices for the theoretical parameter values.  Our empirical findings in ~\cite{MohiBJ09} confirm this assessment in the case of SL0 as well. We have observed fast convergence with excellent empirical recovery rates under weaker sufficient conditions than those that can be obtained from an ARIP analysis.

\bibliography{SepSrc}
\bibliographystyle{IEEEbib}

\end{document}